\documentclass[a4paper,12pt, twoside]{article} 

\usepackage{times}
\usepackage{soul}
\usepackage{url}
\usepackage[hidelinks]{hyperref}
\usepackage[utf8]{inputenc}
\usepackage[small]{caption}
\usepackage{graphicx}
\usepackage{amsthm}
\usepackage{booktabs}
\usepackage{algorithm}
\usepackage{algorithmic}
\usepackage{subfig}
\usepackage{microtype}
\usepackage{xspace}
\usepackage{amsmath}

\usepackage[fixamsmath,disallowspaces]{mathtools}

\usepackage{microtype}

\usepackage{comment}
\usepackage{macros}
\usepackage{multicol}
\usepackage{multirow}

\newtheorem{theorem}{Theorem}
\newtheorem{proposition}[theorem]{Proposition}
\newtheorem{corollary}[theorem]{Corollary}
\newtheorem{example}[theorem]{Example}
\newtheorem{lemma}[theorem]{Lemma}
\newtheorem{claim}[theorem]{Claim}
\newtheorem{definition}[theorem]{Definition}
\newtheorem{remark}[theorem]{Remark}

\usepackage{enumerate}
\usepackage{xcolor}

\usepackage{tikz}
\usetikzlibrary{arrows}
\usepackage{todonotes}

\usepackage{array}
\newcolumntype{H}{>{\setbox0=\hbox\bgroup}c<{\egroup}@{}}

\usepackage{thmtools}
\usepackage{thm-restate}
\usepackage{cleveref}

\usepackage{apptools}
\usepackage{xpatch}
\makeatletter
\xpatchcmd{\thmt@restatable}
{\csname #2\@xa\endcsname\ifx\@nx#1\@nx\else[{#1}]\fi}
{\IfAppendix{\csname #2\@xa\endcsname}{\csname #2\@xa\endcsname\ifx\@nx#1\@nx\else[{#1}]\fi}}
{}{} 
\makeatother
\usepackage{authblk}

\usepackage{pgf,tikz}
\usetikzlibrary{fit,positioning,arrows,shapes, decorations,automata,%
	petri,%
	topaths,calc}
\tikzstyle{dashedarrow} = [-stealth',dashed]
\tikzstyle{tdnode} = [draw,rounded corners,top color=vertexTopColor,bottom color=vertexBottomColor,minimum size=1.5em]
\tikzstyle{stdnode} = [tdnode, font=\scriptsize]
\tikzstyle{stdnodecompact} = [stdnode, inner sep = 1.5pt, outer sep = 0.1pt]
\tikzstyle{stdnodetable} = [stdnode, inner sep = 1.5pt, outer sep = 0]
\tikzstyle{stdnodenum} = [minimum size=1.5em, font=\scriptsize]
\tikzstyle{tdedge} = [-,draw,thick]
\tikzstyle{tdlabel} = [draw=none, rectangle, fill=none, inner sep=0pt, font=\scriptsize]

\newcommand{\futuresketch}[1]{}

\usepackage{xcolor}
\colorlet{vertexTopColor}{white}

\colorlet{vertexBottomColor}{black!10}

\allowdisplaybreaks

\newcommand{\lefttriangle}{\hfill$\triangleleft$}
\AtEndEnvironment{claimproof}{\filledsquare}
\AtEndEnvironment{example}{\lefttriangle}

\usepackage{enumerate,enumitem}

\makeatletter
\providecommand*{\cupdot}{%
	\mathbin{%
		\mathpalette\@cupdot{}%
	}%
}
\newcommand*{\@cupdot}[2]{%
	\ooalign{%
		$\m@th#1\cup$\cr
		\hidewidth$\m@th#1\cdot$\hidewidth
	}%
}
\makeatother

\usepackage{fancyhdr}
\begin{document} 

\title{Inconsistent Databases and Argumentation Frameworks with Collective Attacks}


\author[1]{Yasir Mahmood\thanks{\texttt{yasir.mahmood@uni-paderborn.de}}}
\author[2,3]{Jonni Virtema\thanks{\texttt{jonni.virtema@glasgow.ac.uk}}}
\author[4]{Timon Barlag\thanks{\texttt{barlag@thi.uni-hannover.de}}}
\author[1]{Axel-Cyrille Ngonga Ngomo\thanks{\texttt{axel.ngonga@uni-paderborn.de}}}


\affil[1]{Data Science Group, Heinz Nixdorf Institute, Paderborn University, Germany}
\affil[2]{School of Computing Science, University of Glasgow, United Kingdom}
\affil[3]{School of Computer Science, University of Sheffield, United Kingdom}
\affil[4]{Institut f\"ur Theoretische Informatik, Leibniz Universit\"at Hannover, Germany}

\date{\vspace{-5ex}}  

\maketitle

\begin{abstract}
	The connection between subset-maximal repairs for inconsistent databases involving various integrity constraints and acceptable sets of arguments within argumentation frameworks has recently drawn growing interest. 
	In this paper, we contribute to this domain by establishing a new connection when integrity constraints (ICs) include denial constraints and local-as-view tuple-generating dependencies. 
	It turns out that SET-based Argumentation Frameworks (SETAFs), an extension of Dung's argumentation frameworks (AFs) allowing collective attacks, is needed.
	It is known that subset-maximal repairs under denial constraints correspond to the naive extensions, which also coincide with the preferred and stable extensions in the resulting SETAFs. 
	Our main findings establish that repairs under the considered fragment of tuple-generating dependencies correspond to the preferred extensions. 
	Moreover, for these dependencies, an additional pre-processing allows to compute a unique extension that is stable and naive. 
	Allowing both types of constraints breaks this relationship, and even the pre-processing does not help as only preferred semantics captures these repairs. 
	Finally, while it is known that functional dependencies do not require set-based attacks, we prove the same regarding inclusion dependencies. 
	Thus, one can translate inconsistent databases under these restricted classes of ICs to plain AFs with attacks only between arguments.
\end{abstract}


\section{Introduction}

\newcommand{\an}[1]{\todo[color=yellow]{AN:#1}}

In real-world applications, the provenance of data can be very diverse and include non-trustworthy sources. Thus, databases are often inconsistent in practice when the data does not conform to the imposed integrity constraints. A rich theory has been developed to deal with inconsistent databases. 

Formally, \emph{integrity constraints} (ICs) describe fundamental structural properties of data which should be satisfied for a database to be considered consistent. ICs are often formulated in some logic; most notably in fragments of first-order logic. Two of the most fundamental classes of ICs are key constraints (indicating that the value of a given collection of attributes uniquely determine the whole record in the database table) and inclusion dependencies (that express that all values for a given attribute occur as values of another given attribute). Together key constraints and inclusion dependencies can express foreign key constraints. The literature for database dependencies considers a whole hierarchy of dependencies in increasing generality (see e.g., \cite{DBLP:books/aw/AbiteboulHV95}).  Key dependencies are special cases of functional dependencies (FDs), which are special cases of equality-generating dependencies, which again generalize to denial constraints (DCs). Likewise, inclusion dependencies generalize to local-as-view tuple-generating dependencies (LTGDs), which are special cases of tuple-generating dependencies. All of the aforementioned dependencies can be defined as particular syntactic fragments of first-order logic (see Section \ref{sec:preli} for formal definitions). 
In this paper, we examine expressive integrity constraints, specifically DCs and LTGDs, alongside their less expressive counterparts, FDs and  IDs.

One of the main approaches for handling inconsistency is \emph{database repairing}. The goal is to identify and \emph{repair} inconsistencies in data to obtain a consistent database that satisfies the imposed constraints. In the usual approaches, one would search for a database that satisfies the given constraints and differs minimally from the original database; the obtained database is called a \emph{repair} of the original.
Some of the most prominent notions of repairs are set-based repairs~\cite{tenCate:2012,barcelo2017data}, attribute-based repairs~\cite{Wijsen:2003}, and cardinality-based repairs~\cite{LopatenkoB07}.
Dung's abstract argumentation framework~\cite{Dung95a} has been specifically designed to model conflicts and support relationships among arguments. 
An abstract argumentation framework (AF) represents arguments and their conflicts through directed graphs and allows for a convenient exploration of the conflicts at an abstract level. 
Furthermore, often it is necessary to consider conflicts not only between individual arguments, but between sets of arguments. 
This necessity has given rise to the so-called \emph{set-based argumentation frameworks} (SETAFs) \cite{NielsenP06}, which allow modelling of collective attacks. 

Argumentation frameworks have been explored extensively 
for representing, and reasoning with, inconsistent knowledge bases (KBs), covering Datalog$^{\pm}$ 
and description logics (see e.g., \cite{AriouaC16,AriouaCV17,AriouaTC15,YoungMR16,YunVC20,BienvenuB20} and \cite{ArieliBH19} for an overview). 
The common goal in each of these works is to formally establish a connection between inconsistent KBs and AFs such that the argumentation machinery then outputs extensions equivalent to the set of repairs of the KB. 
Nevertheless, in the setting of relational databases and integrity constraints, there is still a gap with respect to how or whether a connection between inconsistent databases and AFs can be established.
To the best of our knowledge, only functional dependencies (or their generalization, denial constraints) have been investigated in the context of AFs (resp., SETAFs) in~\cite{BienvenuB20}.
We expand this area of research by establishing further connections between repairs and abstract argumentation frameworks when further ICs are allowed. 



In this paper, we focus on subset repairs of relational databases when the ICs are denial constraints and local-as-view tuple-generating dependencies.
The first problem of interest to us is establishing a connection between repairs for the inconsistent DBs and extensions in SETAFs.
We show how subset-maximal repairs under a set of DCs and LTGDs can be obtained by computing the naive, preferred, or stable extensions (see Section~\ref{sec:preli} for definitions) in the related SETAFs.
Repairs under DCs correspond to the naive extensions, which also coincide with the preferred and stable extensions in the resulting SETAFs. 
For LTGDs, the correspondence requires preferred semantics.
However, a polynomial time pre-processing on the resulting SETAFs results in its unique extension that is also stable and naive.
Allowing both types of ICs breaks this relationship between extensions and only preferred semantics captures the repairs.
Our second contribution highlights that when one restricts DCs to FDs and LTGDs to IDs, respectively, a similar connection as for the expressive ICs hold, however, now the resulting SETAFs are in fact AFs (with attacks only between arguments, rather than collective attacks).
We are also interested in the computational problems of deciding the existence of a repair, and determining whether a given fact belongs to some (or every) repair and consider their complexity. 
See Table \ref{table:cont} for the complexity results.

By employing Dung's AFs (and their generalizations, SETAFs) to model repairs of a relational database, one can effectively abstract away from the detailed content of individual entries in the database and focus solely on their relationships with other entries.
This approach provides a clearer understanding of why specific records either appear or do not appear in a repair, as well as the reasons certain values may be absent from query answers.
Furthermore, this modeling approach allows for the incorporation of additional information about records, such as priorities among them, directly at an abstract level. 

\paragraph{Related Work}
The problem of computing subset-maximal repairs and its complexity has been explored extensively in the database setting~\cite{AfratiK09,ArenasBC01,CHOMICKI200590,HannulaW22,LivshitsKR20,StaworkoC10} (see \cite{Bertossi06,Bertossi19} for an overview).
The notions of conflict graphs and hypergraphs have been introduced before for certain families of ICs~\cite{kimelfeld-17,KimelfeldLP20,StaworkoCM12}.
In particular, a correspondence between repairs and subset-maximal independent sets of the conflict graphs for FDs~\cite{ArenasBC01} and DCs~\cite{CHOMICKI200590} has been established. 
Moreover, a recent work defines conflict (hyper)graphs for a richer setting of universal constraints considering symmetric difference repairs~\cite{BienvenuB23}.
Notice that the connection between repairs and independent sets also yields a correspondence between repairs and the naive extensions when the conflict graph is seen as an argumentation framework. 
Nevertheless, to the best of our knowledge, no work has considered a similar graph representation when LTGDs (or even IDs) are taken into account.
Hannula and Wijsen~\cite{HannulaW22} addressed the problem of consistent query answering with respect to primary and foreign keys. 
Their setting allows the insertion of new tuples to fulfill foreign key constraints rather than only deleting.

The case of unirelational databases is also connected to the team-semantics literature~\cite{vaananen07}. Team semantics is a logical framework where formulae are evaluated over unirelational databases (teams in their terminology). 
In this setting, the complexity of finding maximal satisfying subteams has been studied by Hannula and Hella~\cite{HannulaH22} for inclusion logic formulas and by Mahmood~\cite{Mthesis23} for propositional dependence logic. 
In the team-semantics literature, FDs are known as \emph{dependence} atoms and IDs as \emph{inclusion} atoms, denoted respectively as $\depa{x}{y}$ and $\inca{x}{y}$. \footnote{We borrow this notation and write $\depa{x}{y}$ and $\inca{x}{y}$ for FDs and IDs, respectively.}

Our work differs from the previous work, since it combines denial constraints (a subclass of universal constraints) and local-as-view tuple-generating dependencies (a subclass of tuple-generating dependencies).
Moreover, one of our main contributions lies in connecting repairs under DCs (resp., FDs) and LTGDs (IDs) to the extensions of set-based (plain) argumentation frameworks in Dung's setting~\cite{Dung95a}.
A similar connection between (SET)AFs and 
preferred repairs has been explored in the context of prioritized description logic~\cite{BienvenuB20} and Datalog$^{\pm}$ knowledge bases~\cite{AriouaTCB14,CroitoruTV15,CroitoruV13,HoAASA22}.
Employing abstract argumentation, we utilize facts in a database as arguments, aligning with the approach presented in~\cite{BienvenuB20}.
In fact, the results for DCs and FDs regarding preferred and stable semantics follow as a corollary from the work by~\cite{BienvenuB20} since repairs in the prioritized setting generalize subset repairs.
In contrast, Croitoru et al.~\cite{CroitoruTV15,CroitoruV13} take an orthogonal approach in the datalog setting, employing structured argumentation to construct arguments from a given KB.

Regarding the expressivity comparison between various formalisms, \cite{konig2022just} provided several translations between different argumentation formalisms, namely Assumption Based Argumentation (ABA) \cite{bondarenko1997abstract}, Claim-Augmented
Frameworks (CAF)~\cite{dvovrak2020complexity}, and Argumentation Frameworks with Collective Attacks
(SETAF)~\cite{NielsenP06}.
This puts our work in a broader context by presenting novel translations from databases with expressive integrity constraints to SETAFs.
Furthermore, recently \cite{DBLP:conf/aaai/0002HN25} presented translations from AFs to inconsistent databases which is the converse direction to what we aimed to achieve in this paper. 
Their main findings indicate that one can translate argument interaction in an AF via functional and inclusion dependencies alone.

\paragraph{Prior Work} A preliminary version of this work was published in FoIKS 2024~\cite{mahmood2024computing}.
The current paper expands the earlier work to more expressive integrity constraints.
To be precise, we utilize SETAFs to simulate denial constraints and LTGDs using the same ideas as for functional and inclusion dependencies. The conference version focussed on functional and inclusion dependencies in connections to plain AFs.
The results for DCs follow from the work of \cite{BienvenuB20}, thus the task is to extend them to LTGDs.
The results of the conference version 
can now be derived as special cases from the more general results of this paper, when the involved ICs are FDs and IDs.

\begin{table}[t]
	\centering
	\scalebox{0.8}{
		\begin{tabular}{l@{\; }c@{\; }c@{\; }c@{\; }c@{\; }c@{\; }} 
			\toprule
			\multicolumn{6}{c}{Equivalence between DBs with \textbf{expressive ICs} and SETAFs, along with \textbf{data complexity} results.} \\
			\midrule
			& ICs & SETAF-equivalent semantics &  \multicolumn{3}{c}{Complexity Results} \\
			& &\multicolumn{1}{c}{for $\repairs$} & $\REPB$ & $\somerepairB$ & $\allrepairB$  \\ \midrule 
			& DCs & $\sigma \in \{\naive,\pref,\stab\}^\star$ 
			& (trivial) & (trivial) & $\in\Ptime$ \\
			& LTGDs & $\pref$ \scriptsize{(Thm. \ref{thm:ext-tgds})} & $\in\Ptime$\textsuperscript{\cite{AfratiK09}} & $\in\Ptime$\textsuperscript{\cite{AfratiK09}} & $\in\Ptime$\textsuperscript{\cite{AfratiK09}} \\
			& DCs+LTGDs & $\pref$ \scriptsize{(Thm. \ref{thm:ext-both-setafs})} & $\NP$ \scriptsize{(Thm. \ref{thm:cred-both-extended})} & $\NP$ \scriptsize{(Thm. \ref{thm:cred-both-extended})} & $\PiP$  \scriptsize{(Thm. \ref{thm:skep-both-extended})} \\
			\\ 
			\toprule 
			\multicolumn{6}{c}{Equivalence between DBs with \textbf{less expressive ICs} and AFs, along with \textbf{combined complexity} results.} \\
			\midrule
			
			& ICs & AF-equivalent semantics &  \multicolumn{3}{c}{Complexity Results} \\
			& &\multicolumn{1}{c}{for $\repairs$} & $\REP$ & $\somerepair$ & $\allrepair$  \\ \midrule 
			
			& FDs & $\sigma \in \{\naive,\pref,\stab\}^\star$ 
			& (trivial) & (trivial) & $\in\Ptime$ \\
			& IDs & $\pref$ \scriptsize{(Cor. \ref{thm:ext-inc})} & $\in\Ptime$\textsuperscript{\cite{AfratiK09}} & $\in\Ptime$\textsuperscript{\cite{AfratiK09}} & $\in\Ptime$\textsuperscript{\cite{AfratiK09}} \\
			& FDs+IDs  & $\pref$ \scriptsize{(Cor. \ref{thm:ext-both})} & $\NP$ \scriptsize{(Thm. \ref{thm:cred-both-fixed})} & $\NP$ \scriptsize{(Thm. \ref{thm:cred-both-fixed})} & $\PiP$  \scriptsize{(Thm. \ref{thm:skep-both})} \\
			\bottomrule
	\end{tabular}}
	%
	\caption{Overview of our main contributions. The complexity results depict completeness, unless specified otherwise. 
		The table at the top indicates results for 
		expressive ICs involving a fixed set $B$ of DCs and LTGDs, 
		whereas the table at the bottom depicts results for less expressive ICs including FDs and IDs.
		{The lower bounds for less expressive ICs also hold for data complexity since the involved reductions use fixed sets of constraints.}
		The second column in each table indicates the (SET)AF-semantics corresponding to repairs for ICs in the first column, and the last three columns present the complexity of each problem. The $\Ptime$-results are already known in the literature, whereas the remaining results are new. Finally, the results marked by $\star$ follow from the earlier work by \cite{BienvenuB20} and \cite{CHOMICKI200590}.}\label{table:cont}
\end{table}

\section{Preliminaries}\label{sec:preli}
We assume that the reader is familiar with basics of complexity theory and first-order logic.
In particular, we will encounter (1) the complexity classes $\Ptime, \NP$, and $\PiP$, and (2) closed first-order formulas having $\forall.\exists.$ as the quantifier prefixes.
In the following, we shortly recall the necessary definitions from databases and argumentation.

\paragraph{Databases}
A (relational) \emph{schema} $\tau$ is a finite set of relation names. Each relation $T\in\tau$ is associated with a positive integer, called the \emph{arity} of $T$.
 {We consider a countable set $C$ of constants to act as domain elements, from which the values can appear in a database.
	A \emph{term} is either a constant from $C$ or a variable.
	Let $t_1,\dots,t_k$ be terms and $T\in\tau$ be a relation  of arity $k$, the expression $T(t_1,\dots,t_k)$ is called a \emph{relational atom}. 
}A relational atom composed of only  constants is called a \emph{fact}.
Given a schema $\tau$, a \emph{database} $\calT$ over $\tau$ is a finite set of facts using relation symbols from $\tau$.
The \emph{active domain} of a database $\calT$ is denoted as $\dom(\calT)$ and defined as the collection of all the constant values from $C$ that occur in facts of $\calT$.

A fact can also be seen as a database tuple.
 {In this notation, the fact $T(a_1,\dots,a_k)$ corresponds to stating that the tuple $(a_1,\dots,a_k)$ is a record in the \emph{(database) table} $T$ of arity $k$---which is a collection of such tuples, each with an associated unique identifier}.
For analogy, a table corresponds to a relation name $T$ in a schema $\tau$ and 
can be seen as a grouping of facts together that use the same relation symbol (namely, $T$).
 {For a table $T$ of arity $k$, the positions $\{1,\dots,k\}$ in the columns of $T$ are usually given explicit names, resulting in the set $\att(T)$ of attributes of $T$.
	The active domain $\dom(T)$ of $T$ in the tabular view is the collection of all the values that occur in the tuples of $T$.}
We denote individual attributes (or variables) by small letters (e.g., $x$, $y$) and reserve boldface letters (e.g., $\tuple x, \tuple y$) for sequences of attributes.
For an attribute $x\in\att(T)$ and a tuple $s\in T$, $s(x)$ denotes the value taken by $s$ for the attribute $x$. 
For a sequence $\tuple x =(x_1,\ldots,x_k)$, $s(\tuple x)$ denotes the sequence of values $(s(x_1),\ldots, s(x_k))$.
Given this tabular view, a {database} over schema $\tau = \{T_1,\ldots, T_m \}$ is seen as a collection $\calT=(T_1,\dots, T_m)$ of tables corresponding to relation symbols in $\tau$.
 {Again, the active domain $\dom(\calT)$ of a database $\calT$ is the union of active domains of each table $T\in\calT$.} 
We find it convenient to write a database without explicitly mentioning its schema and active domain always.


%
We next present an example database.

\begin{example}
	\label{ex:running_1}
	Consider the schema $\{E,D\}$. 
	The database tables over this schema (Figure~\ref{tab:running_1}) contain data about employees and departments, respectively.
	 {Here, $E(\text{E1}, \text{D1}, \text{Paderborn})$ is a fact and $(\text{E1}, \text{D1}, \text{Paderborn})$ is a tuple in $E$ with identifier $e_1$. For presentation of examples, we refer to tuple identifiers also as facts. Hence, we write $e_1$ for the fact $E(\text{E1}, \text{D1}, \text{Paderborn})$.}
	\begin{table}
		\centering
		\begin{minipage}{0.45\textwidth}
			\begin{tabular}{l@{\hskip 5pt}|@{\hskip5pt}c@{\;}ccc@{\hskip5pt}}
				$E$ & Emp\_ID  & Dept\_ID & Location\\ \hline
				$e_1$ & E1 & D1 & Paderborn\\
				$e_2$ & E2 & D2 & Sheffield\\
				$e_3$ & E3 & D2 & Hanover
			\end{tabular}
		\end{minipage}
		\begin{minipage}{0.45\textwidth}
			\begin{tabular}{l@{\hskip 5pt}|@{\hskip5pt}c@{\;}ccc@{\hskip5pt}}
				$D$ & Dept\_ID  & Dept\_Name & Location\\ \hline
				$d_1$ & D1 & Sales & Paderborn\\
				$d_2$ & D2 & Marketing & Sheffield\\
				$d_3$ & D3 & HR & Hanover
			\end{tabular}
		\end{minipage}
		\caption{A database with two tables over the schema $\{E,D\}$. The first column indicates identifier for each fact/tuple.}
		\label{tab:running_1}
	\end{table}
\end{example}

\paragraph{Integrity Constraints.}
Integrity constraints are closed first-order formulas, i.e., all the variables are bound to quantification.
In the following, we first specify different types of ICs.

Let  $\tau$ be a database schema.
Recall that a {database atom} is an application of any relation name $T\in\tau$ over variables or constants.
Let $\varphi$ and $\psi_i$ be conjunctions of database atoms and let $\beta$ be a quantifier-free formula using only built-in predicates (e.g., equality or inequality).
We call a formula of the form 
\[
\forall \tuple{x} (\varphi(\tuple{x})\land \beta(\tuple x)\rightarrow \bigvee_{i=1}^n \exists \tuple{y}_i \psi_i(\tuple x, \tuple {y}_i))
\]
a generic \emph{first-order integrity constraint} (FO).
All other types of constraint languages arise from restrictions on the generic ICs, as we illustrate next.
A generic IC is called \emph{full} or a \emph{universal constraint} (UC) if it contains no existential quantifiers.
A \emph{denial constraint} (DC) has the form $\forall \tuple x \neg (\varphi(\tuple x)\land \beta (\tuple x)) $, which can also be thought of as a UC with empty right hand side.
An \emph{equality-generating dependency} (EGD) has the form $\forall \tuple x (\varphi(\tuple x)\rightarrow x_i=x_j)$, which can be thought of as a DC where $\beta$ is a single inequality.
Let $T$ be an $n$-ary relation and $I,J\subseteq \{1,\dots,n\}$ be two sets with pairwise distinct variables $x_1,\dots,x_n,y_1,\dots,y_n$.
Then, a \emph{functional dependency} (FD) over $T$ is an EGD of the form $\forall \tuple x \forall \tuple y (T(\tuple x) \land T(\tuple y) \land \bigwedge_{i\in I}x_i =y_i \rightarrow \bigwedge_{j\in J}x_j=y_j)$.
A \emph{key constraint} is a special type of FD requiring $I\cup J = \{1,\dots,n\}$.

A \emph{disjunctive tuple-generating dependency} ($\vee$-TGD) is a generic IC that has empty $\beta$ and an ordinary \emph{tuple-generating dependency} (TGD) additionally requires $n=1$.
A \emph{local-as-view} TGD (LTGD) is a TGD where $\varphi$ is a single atom and an \emph{inclusion dependency} (ID) is an LTGD where $\psi_i$ is also a single atom.
Figure~\ref{fig:ICs} (adapted from \cite{arming2016complexity}) presents an overview of the hierarchy of different ICs together with their syntactic forms.

The semantics for ICs is defined similarly to first-order (FO) formulas.
 {A database $\calT$ over schema $\tau$ gives rise to an FO-structure $\mathcal A$ over the vocabulary $\tau$. The domain of $\mathcal A$ is the active domain $\dom(\calT)$ of $\calT$ 
	and each relation symbol $T\in \tau$ is interpreted as the corresponding database table $T$ in $\calT$, i.e., it contains all the tuples of constants $(a_1,\dots,a_k)$ over $\dom(\calT)$ such that $T(a_1,\dots,a_k)\in \calT$.
	Then, for an IC $\alpha$, we write $\calT\models \alpha$ iff $\mathcal A\models \alpha$ under the classical Tarski semantics.}
Moreover, one often uses a more convenient notation for FDs and IDs as we illustrate next.
Let $T$ be an $n$-ary relation, $I,J\subseteq \{1,\dots,n\}$ be two sets and let $x_1,\dots,x_n,y_1,\dots,y_n$ be pairwise distinct variables.
Then, for sequences $\tuple X = (x_i \mid i\in I)$ and $\tuple Y = (y_j \mid j\in J)$, the {functional dependency} $\forall \tuple x \forall \tuple y (T(\tuple x) \land T(\tuple y) \land \bigwedge_{i\in I}x_i =y_i \rightarrow \bigwedge_{j\in J}x_j=y_j)$ can be equivalently expressed as $\depa{X}{Y}$ or $\tuple X \rightarrow \tuple Y$ (as done usually in the DBs setting).
Moreover, let $T_i$ and $T_j$ be two relations, and $\tuple x_1, \tuple x_2, \tuple x_3$ be sequences of variables.
The inclusion dependency $\forall \tuple x_1 \forall \tuple x_2 (T_i(\tuple x_1,\tuple x_2) \rightarrow \exists \tuple x_3 T_j(\tuple x_2,\tuple x_3))$ can be equivalently expressed as $\incageneral{T_i}{Y}{T_j}{X}$ where $\tuple X$ (respectively, $\tuple Y$) is the set of attributes of $T_j$ ($T_i$) corresponding to $\tuple x_2$.
We find it convenient to write $\tuple Y\subseteq \tuple X$ instead of $T[\tuple Y]\subseteq T[\tuple X]$ for an ID involving a single relation $T$.
We will use this notation throughout our paper.

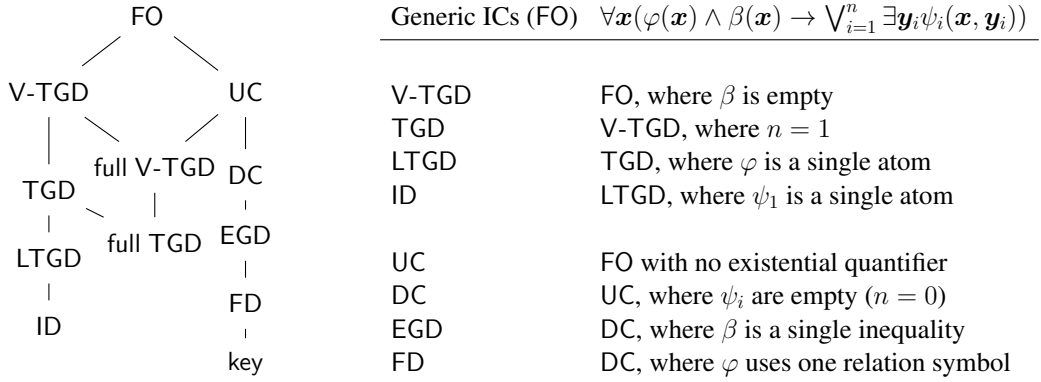
\begin{figure}
	\centering
	\begin{tikzpicture}[level distance=1.7em, sibling distance=5em, minimum height=1.75em, level 2/.style={sibling distance=3em},
		every node/.style={scale=.9, font=\small},
		edge from parent/.style={thick,-,black, draw}]
		\node (fo) at (0,-.2)  {$\FO$};
		\node (vtgd) at (-1.3,-1.2) {$\vtgd$};
		\node (UC) at (1.3,-1.2)  {$\UC$};
		\node (tgd) at (-1.3,-2.5) {$\tgd$};
		
		\node (fullvtgd) at (0.1,-2.2) {$\fullvtgd$};
		\node (fulltgd) at (0.1,-3.2) {$\fulltgd$};
		
		\node (lavtgd) at (-1.3,-3.4) {$\lavtgd$};
		\node (id) at (-1.3,-4.3) {$\id$};
		
		\node (denial) at (1.3,-2.3) {$\denial$};
		\node (egd) at (1.3,-3.1) {$\egd$};
		\node (fd) at (1.3,-4) {$\fd$};
		\node (key) at (1.3,-4.8) {$\key$};
		
		\foreach \f/\g in {fo/vtgd, fo/UC, UC/denial, denial/egd, egd/fd, fd/key, UC/fullvtgd, fullvtgd/fulltgd, vtgd/fullvtgd, tgd/fulltgd, vtgd/tgd, tgd/lavtgd,lavtgd/id} {
			\draw[-] (\f) -- (\g);
		}
	\end{tikzpicture}
	\hspace{1cm}
	\begin{tikzpicture}[level distance=1.7em, sibling distance=5em, minimum height=1.75em, level 2/.style={sibling distance=3em},
		every node/.style={scale=.7, font=\large},
		edge from parent/.style={thick,-,black, draw}]
		
		\node (uc) at (0,0)  {
			\begin{tabular}{l l}
				Generic ICs ($\FO$) & $\forall \tuple{x} (\varphi(\tuple{x})\land \beta(\tuple x)\rightarrow \bigvee_{i=1}^n \exists \tuple{y}_i \psi_i(\tuple x, \tuple {y}_i))$\\\midrule \\
				$\vtgd$   & $\FO$, where $\beta$ is empty \\ 
				$\tgd$    & $\vtgd$, where $ n=1$ \\ 
				$\lavtgd$ & $\tgd$, where $\varphi$ is a single atom \\
				$\id$ & $\lavtgd$, where $\psi_1$ is a single atom \\
				\\
				
				$\UC$     & $\FO$ with no existential quantifier \\
				$\denial$ & $\UC$, where $ \psi_i$ are empty ($n=0$)\\
				$\egd$    & $\denial$, where $\beta$ is a single inequality \\
				$\fd$ &  $\denial$, where $\varphi$ uses one relation symbol
				
		\end{tabular}};
	\end{tikzpicture}
	\caption{Hierarchy of ICs and syntactic form for most commonly studied constraints~\cite{arming2016complexity}. 
	}
	\label{fig:ICs}
\end{figure}

The following example depicts ICs over the database from Example~\ref{ex:running_1}. 

\begin{example}
	\label{ex:running_2}
	Consider a DC $\mathrm{dc}$ and an LTGD $\mathrm{lav}$, as given below.
	\[
	\mathrm{dc} \coloneqq \forall x_1, x_2, x_3, x_4, x_5 \neg (E(x_1, x_2, x_3) \land D(x_2, x_4, x_5) \land x_3 \neq x_5)
	\]
	\[
	\mathrm{lav} \coloneqq \forall x_1, x_2, x_3 (D(x_1, x_2, x_3) \to \exists y_1, y_2 E(y_1, x_1, y_2))
	\]
	The denial constraint $\mathrm{dc}$ states that the database should not contain any employees that work at a different location than where their department is located and the LTGD $\mathrm{lav}$ requires that every department has at least one employee.
	Both $\mathrm{dc}$ and $\mathrm{lav}$ are violated in the database from Example~\ref{ex:running_1}.
\end{example}

Let $\calT$ be a database (a collection of facts) and $B$ be a collection of ICs.
Then $\calT$ is \emph{consistent} with respect to $B$, denoted as $\calT\models B$, if $\calT\models b$ for each $b\in B$.
Moreover, $\calT$ is \emph{inconsistent} with respect to $B$ if there is some $b\in B$  such that $\calT\not \models b$.
A \emph{subset-repair} of $\calT$ with respect to $B$ is a subset $\calP\subseteq \calT$ that is consistent with respect to $B$, and maximal in the sense that no set $\calP'$ exists such that it is consistent with respect to $B$ and $\calP \subset \calP' \subseteq T$.
In the tabular view, the subset of a DB $\calT = (T_1,\dots,T_m)$ is defined as $\calP = (P_1,\dots,P_m)$, where $P_i\subseteq T_i$ for each $i\leq m$.
In the following, we simply speak of a repair when we intend to mean a subset-repair.
Furthermore, we often consider a database $\calT$ and a set $B$ of ICs assuming that $B$ only uses relation names from the schema of $\calT$.
In that case, we call $\calD = \langle \calT, B \rangle$ a \emph{constrained database}.
For a constrained database $\calD$, $\repairs(\calD)$ denotes the set of all repairs for $\calD$.

The first problem we are interested in ($\REP$) asks whether there exists a repair 
for a constrained database $\calD$.
 {Since the empty database ($\emptyset$) satisfies each IC trivially, we are interested in the decision problem asking whether a non-empty repair exists.
}
%
\problemdef{$\REP$}{a constrained database $\calD=\langle \calT, B\rangle$}{ {is there a repair $\mathcal R\in \repairs(\calD)$ with $\mathcal R\neq \emptyset$}}

Two further problems of interest are 
\emph{brave} and \emph{cautious} reasoning for a given fact $s\in \calT$, asking whether $s$ belongs to some (every) repair for $\calD$.

\problemdef{$\somerepair$}{a constrained database $\calD=\langle \calT, B\rangle$ and a fact $s\in \calT$}{does $s$ belong to some repair for $\calD$}

\problemdef{$\allrepair$}{a constrained database $\calD=\langle \calT, B\rangle$ and a fact $s\in \calT$}{does $s$ belong to all repairs for $\calD$}

We also consider the case when the set $B$ of integrity constraints is fixed and the input only involves a database $\calT$ (and a fact $s$).
 {This gives rise to separate decision problems $\REPB$, $\somerepairB$ and $\allrepairB$, for each $B$. We can then define the of \emph{data complexity} of $\REP$, $\somerepair$ and $\allrepair$, in the usual way. That is, for a complexity class $\mathcal{C}$, the data complexity of  $\REP$ is in $\mathcal{C}$, if $\REPB$ is in $\mathcal{C}$, for every $B$, and $\mathcal{C}$-hard, if $\REPB$ is $\mathcal{C}$-hard, for some $B$.}
Thus, the problems defined above (i.e., $\REP$, $\somerepair$ and $\allrepair$) then correspond to the setting of \emph{combined complexity}.

Observe that our focus in this work lies on inconsistent databases, i.e., we assume that $\calT$ is inconsistent with respect to $B$ in a constrained database $\calD=\langle \calT, B\rangle$. 
This is not a restriction since the considered questions about repairs are trivial if the database is consistent with respect to its constraints.
Moreover, all our translations to (SET)AFs and the connection between repairs and extensions still apply to the consistent case, although they do not provide much insight. 

The following example presents a repair for the database from Example~\ref{ex:running_1} with respect to the ICs from Example~\ref{ex:running_2}. 
\begin{example}
	\label{ex:running_3}
	Reconsider the database from Example~\ref{ex:running_1}, which is inconsistent with the ICs from Example~\ref{ex:running_2}.
	Here, the fact $d_3$ violates $\mathrm{lav}$, and the only way to resolve this violation (in our subset-repair setting) is to remove $d_3$.
	Moreover, the set $\{e_3, d_2\}$ also violate $\mathrm{dc}$.
	This results in two ways to repair our database, by removing either of those facts.
	Thus, we have two subset-maximal repairs, namely $\{e_1, e_2, d_1, d_2\}$ and $\{e_1, e_2, e_3, d_1\}$.
\end{example}

\paragraph*{Abstract Argumentation}
We use Dung's argumentation framework~\cite{Dung95a} and consider only non-empty and finite sets of arguments~$A$.
An \emph{(argumentation) framework~(AF)} is a directed graph~$\calF=(A, R)$, where $A$ is a set of arguments and the relation $R \subseteq A\times A$ represents direct attacks between arguments.
If $S\subseteq A$, we say that an argument~$s \in A$ is \emph{defended by $S$ in $\calF$}, if for every $(s', s) \in R$ there exists $s'' \in S$ such that $(s'', s') \in R$.

In abstract argumentation one is interested in computing the so-called \emph{extensions}, which are subsets~$S \subseteq A$ of the arguments that have certain properties.
The set~$S$ of arguments is called \emph{conflict-free in~$\calF$} if $(S\times S) \cap R = \emptyset$.
Let $S$ be conflict-free, then $S$ is
\begin{enumerate}
	\item \emph{naive in $\calF$} if no $S' \supset S$ is \emph{conflict-free} in $\calF$;
	\item \emph{admissible in $\calF$} if every $s \in S$ is \emph{defended by $S$ in $\calF$}.
	
\end{enumerate}

Further, let 
$S$ be admissible. Then, $S$ is
\begin{enumerate}
	\setcounter{enumi}{2}
	\item \emph{preferred in~$\calF$}, if there is no $S' \supset S$ that is \emph{admissible in $\calF$};
	\item \emph{stable in~$\calF$} if every $s \in A \setminus S$ is \emph{attacked} by some $s' \in S$.
\end{enumerate}

We denote each of the mentioned semantics by abbreviations: $\conf, \naive,$ $\adm, \pref,$ and $\stab$, respectively.
For 
a semantics~$\sigma \in \{\conf, \naive, \adm, \pref,  \stab\}$, we write $\sigma(\calF)$ for the set of \emph{all extensions} of semantics~$\sigma$ in $\calF$. 
Now, we are ready to define the corresponding decision problem asking for extension existence with respect to a semantics $\sigma$.
\problemdef{$\sem\sigma $}{an argumentation framework~$\calF$}{is it true that $\sigma(\calF)\neq\emptyset$}

Finally, for an AF~$\calF{=}(A,R)$ and $a\in A$, we define the problems $\cred_{\sigma}$ and $\skep_{\sigma}$, which ask whether~$a$ is in some $\sigma$-extension of $\calF$ (``\emph{credulously} accepted'') or every $\sigma$-extension of~$\calF$  (``\emph{skeptically} accepted''), respectively. 

\problemdef{$\cred_\sigma $}{an AF~$\calF = (A, R)$ and an argument $a \in A$}{is it true that $a \in E$ for some $E \in \sigma(\calF)$}

\problemdef{$\skep_\sigma $}{an  AF~$\calF = (A, R)$ and an argument $a \in A$}{is it true that $a \in E$ for all $E \in \sigma(\calF)$}


\begin{example}\label{ex:af}
	Consider the argumentation framework $\mathcal{AF} = \{a,b,c,d\}$ with the attack relation as depicted in Figure~\ref{fig:af}.
	Then extensions for $\calF$ are also depicted in Figure~\ref{fig:af} for the mentioned semantics.
	\begin{figure}
		\centering
		\begin{minipage}{0.4\textwidth}
			\begin{tikzpicture}[scale=.8,arg/.style={circle,draw=black,fill=white,inner sep=.75mm}]
				\node[arg] (a) at (0,1.5) {a};
				\node[arg] (b) at (2,1.5)  {b};
				\node[arg] (c) at (2,0)  {c};
				\node[arg] (d) at (0,0)  {d};
				
				\foreach \f/\t in {a/b,b/a,b/c,c/d}{
					\path[-stealth',draw=black, thick] (\f) edge (\t);
				}
			\end{tikzpicture}
		\end{minipage}
		\begin{minipage}{0.45\textwidth}
			\[
			\begin{array}{l l}
				\text{Naive:} & \{a,c\}, \{a,d\}, \{b, d\} \\
				\text{Admissible:} & \{a\}, \{a,c\}, \{b\}, \{b,d\} \\
				\text{Preferred:} & \{a,c\}, \{b,d\} \\
				\text{Stable:} & \{a,c\}, \{b,d\}
			\end{array}
			\]
		\end{minipage}
		\caption{Framework $\calF$ from Example~\ref{ex:af} (Left) and  extensions for given semantics (Right).}
		\label{fig:af}
	\end{figure}
	%
\end{example}

The complexity of reasoning in argumentation is well understood, see~\cite[Table 1]{DvorakDunne17} for an overview.
In particular, $\cred_{\naive}$ and $\skep_{\naive}$ are in $\Ptime$, 
whereas, $\cred_{\pref}$  and $\skep_{\pref}$ are $\NP$-complete and $\PiP$-complete, respectively.
Moreover, the problem 
to decide whether there is a non-empty extension is in $\Ptime$ for $\naive$ and $\NP$-complete for $\pref$-semantics.
This makes $\naive$-semantics somewhat easier 
and $\pref$ the hardest among the considered semantics in this work.

\paragraph{Set-Based Argumentation.}
Dung's AFs can be generalized by allowing attacks not only from single arguments, but from collections of arguments \cite{NielsenP06}.
A \emph{B-hypergraph} (backwards hypergraph) is a directed hypergraph $(V, E)$ where each edge goes from a set of nodes to a single node, i.e., $E \subseteq 2^V \times V$.
A \emph{set-based argumentation framework (SETAF)} is a B-hypergraph $\calF = (A, R)$ where $A$ is the set of arguments and $R \subseteq 2^A \times A$ represents the attacks from sets of arguments to individual arguments. 
A subset $S \subseteq A$ of arguments is \emph{conflict-free}, if $(2^S \times S) \cap R = \emptyset$ and $S$ \emph{defends} an argument $s \in A$ if for every $(S', s) \in R$ there exists $S'' \subseteq S$ such that $(S'', s') \in R$ for some $s' \in S'$.
For $S=\{s\}$ a singleton set, we prefer denoting the attack $(\{s\},s')$ by $(s, s')$.
With this definition of conflict-freeness, the notions of \emph{naivety}, \emph{admissibility}, \emph{preferredness} and \emph{stability} are defined analogously to plain AFs.

The corresponding decision problem for the extensions existence is defined as follows.
\problemdef{$\sem\sigma^{\textit{SET}}$}{a SETAF $\calF$}{is it true that $\sigma(\calF) \neq \emptyset$}

For a SETAF~$\calF$ and argument $a\in A$, we similarly define the problems $\cred_{\sigma}$ and $\skep_{\sigma}$ asking whether~$a$ is in some $\sigma$-extension of $\calF$ or every $\sigma$-extension of~$\calF$, respectively. 
\begin{example}
	\label{ex:setaf}
	Consider the SETAF $\calS = (A, R)$, where $A = \{a, b, c, d, e\}$ and $R = \{(\{a, b\}, d), (\{b, c\}, e), (\{e\}, a)\}$.
	In $\calS$, the extension $\{a, b, c\}$ is conflict-free, preferred and stable. Note that $\{b, c\}$ defends $a$ against $\{e\}$.
	This example is visualized in Figure~\ref{fig:setaf}.
	\begin{figure}
		\begin{center}
			\begin{tikzpicture}[scale=.9,arg/.style={circle,draw=black,fill=white,inner sep=.75mm}]
				\node[arg] (a) at (1,3) {$a$};
				\node[arg] (b) at (3,3.5) {$b$};
				\node[arg] (c) at (5,3) {$c$};
				\node[arg] (e) at (2,1) {$e$};
				\node[arg] (d) at (4,1) {$d$};
				
				\draw[rounded corners=6pt, thick, rotate=10] (1, 2.4) rectangle (4, 3.34);
				\draw[rounded corners=6pt, thick, rotate=350] (1.9, 3.4) rectangle (5, 4.35);
				
				\draw[-stealth',draw=black, thick ] (2, 2.79) -> (d);
				\draw[-stealth',draw=black, thick ] (4, 2.75) -> (e);
				\draw[-stealth',draw=black, thick ] (e) -> (a);
			\end{tikzpicture}
		\end{center}
		
		\caption{The SETAF $\calS$ from Example~\ref{ex:setaf}.}
		\label{fig:setaf}
	\end{figure}
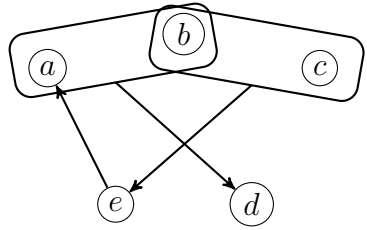
\end{example}

\section{Translating Databases to Argumentation Frameworks}

In this section, we prove that repairs under expressive families of integrity constraints (involving DCs and LTGDs) can be equivalently seen as extensions in SETAFs.
In the first two subsections, we consider instances containing only one type of ICs to SETAFs and the third subsection combines both constraints (DCs and LTGDs).
Given an instance $\calD = \langle \calT,B\rangle $ comprising a database $\calT$ and a set $B$ of ICs, the goal is to capture all the subset-repairs for $\calD$ by $\sigma$-extensions of the resulting SETAF $\mathcal S_\calD$ for some semantics $\sigma$.

In Section~\ref{sec:dcs}, we encode a constrained database $\mathcal D = \langle \calT,D \rangle $ with database $\calT$ and a collection $D$ of DCs into a SETAF $\mathcal S_\calD$.
This is achieved by letting each fact $s\in \calT$ be an argument.
Then the attack relation between arguments simulates the violation of some $d \in D$ by sets of facts. 
This allows to establish a connection between repairs for $\calD$ and the extensions for the resulting SETAF under preferred, stable and naive semantics.
This construction for DCs is a special case of the results by \cite{BienvenuB20}.
In our case, we do not consider priorities among facts in the database which allows a weaker SETAF-semantics (naive extensions) to capture repairs.

In Section~\ref{sec:tgds}, we simulate 
a constrained database $\calD= \langle \calT, L\rangle$ including a collection $L$ of LTGDs via SETAFs.
The first observation is that the semantics for LTGDs requires the notion of \emph{support} or \emph{defense} rather than \emph{conflict} between facts.
Then, we depict each fact as an argument as well as use auxiliary arguments to simulate the dependency $\ell \in L$ (i.e., to model the semantics for LTGDs).
Further, we add self-attacks for these auxiliary arguments to prohibit them from appearing in any extension.
Consequently, we establish a connection between repairs for $\calD$ and the extensions for SETAFs under preferred semantics. 
Finally, we present a polynomial-time processing of the resulting SETAFs that computes its unique preferred extension, which also yields a unique repair for $\calD$.
Interestingly, the computed preferred extension is also stable and naive, and hence coincides with the only repair for $\calD$.
Our construction for LTGDs relies on self-attacking arguments.
This allows us to avoid taking auxiliary arguments in any extension and establish a connection between repairs of a DB and extensions of the resulting SETAF.
Towards the end of our paper, we discuss (Section~\ref{sec:self-attacks}) the role played by the self-attacking arguments and a possibility for a different translation from DBs to (SET)AFs.

Having established that both DCs and LTGDs can be modeled in SETAFs via attacks, Section~\ref{sec:dcs-tgds} generalizes this approach by allowing both types of constraints. 
%
We observe that no prior work has handled the case of tuple-generating dependencies (to our knowledge).

\subsection{Simulating Denial Constraints via SETAFs}\label{sec:dcs}

We first formalize the notion of conflicts for a DB with respect to a set of denial constraints.
Intuitively, a \emph{conflict} in a constrained database $\calD$ is a minimal (under set inclusion) inconsistent set of facts from $\calT$ with respect to some denial constraint in $\calD$.

\begin{definition}
	 {Let $\calD=\langle\calT, D \rangle$ be a constrained database 
		with a set $D$ of DCs.
		A collection 
		$C\subseteq \calT$ of facts is a \emph{conflict} in $\calD$ if there exists some $d\in D$ such that (i)  
		$C\not\models d$, and 
		(ii) $C'\models d$ for each proper subset $C'\subsetneq C$.}
\end{definition}

Given $\calD$, then by $\conflicts(\calD)$ we denote the set of all conflicts in $\calD$.
It is known that $\conflicts(\calD)$ can be represented as a \emph{conflict hypergraph}~\cite{CHOMICKI200590} where the vertices are the facts in $\calD$ and hyperedges are conflicts in $\calT$.
We follow \cite{BienvenuB20} and transform an instance $\calD = \langle \calT, D\rangle$ with database $\calT$ and a collection $D$ of DCs to a SETAF $\calS_\calD$.

\begin{definition}[\cite{BienvenuB20}]\label{def:setaf-dc}
	Let $\calD = \langle \calT, D\rangle$ be a constrained database including a database $\calT$ and a collection $D$ of DCs. 
	Then, $\calS_\calD$ denotes the following SETAF.
	\begin{itemize}
		\item $A\dfn \calT$, i.e, each fact $s\in \calT$ is seen as an argument.
		\item $R\dfn \{(C\setminus\{t\}, t) \mid C\in\conflicts(\calD), t\in C\}$.
	\end{itemize}
	We call $\calS_\calD$ the SETAF generated by $\calD$.
\end{definition}

It is worth highlighting that given a set $D$ of denial constraints as input, then checking whether ``a database $\calT'$ is a repair of a database $\calT$ w.r.t. $D$'' is $\DP$-complete even if $\calT$ and $\calT'$ are both fixed~\cite[Lemma 6]{arming2016complexity}.
The complexity of the repair checking problem drops to $\Ptime$ if we consider a {fixed set} of denial constraints.
As a result, one normally considers a fixed set $D$ of DCs.
In our encoding, this has the effect that one can construct the argumentation framework from a constrained database $\calD$ in polynomial time.
Furthermore, given the conflict hypergraph $\calG_\calD$ for a constrained database $\calD$, the framework $\calS_\calD$ can still be constructed in polynomial time in the size of $\calG_\calD$.

\begin{example}
	\label{ex:setaf_dcs}
	Consider the following database tables $C$ and $O$, containing the data about customers and orders, respectively.
	\begin{center}
		\begin{tabular}{l@{\hskip 5pt}|@{\hskip5pt}c@{\;}cc@{\hskip5pt}}
			$C$ & Cus\_ID  & Location\\ \hline
			$s_1$ & C1 & Paderborn\\
			$s_2$ & C2 & Sheffield\\
			$s_3$ & C3 & Hanover\\
		\end{tabular}
		\hspace{5mm}
		\begin{tabular}{l@{\hskip 5pt}|@{\hskip5pt}c@{\;}cccc@{\hskip5pt}}
			$O$ & Ord\_ID  & Cus\_ID & Prod\_ID & Ship\_ID\\ \hline
			$t_1$ & O1 & C1 & P1 & S1\\
			$t_2$ & O2 & C2 & P2 & S2\\
			$t_3$ & O3 & C2 & P1 & S1\\
			$t_4$ & O5 & C1 & P1 & S4\\
		\end{tabular}
	\end{center}
	Here, Cus\_ID, Ord\_ID, Prod\_ID and Ship\_ID are the IDs of individual customers, orders, products and shipments, respectively.
	Let $D = \{\mathrm{dc}_1, \mathrm{dc}_2\}$ be a set of denial constraints where
	\[
	\mathrm{dc}_1 \coloneqq \forall x_1, \dots , x_6 \neg (O(x_1, x_2, x_3, x_4) \land O(x_5, x_2, x_3, x_6) \land x_1 \neq x_5)
	\]
	and
	\[
	\mathrm{dc}_2 \coloneqq \forall x_1, \dots, x_9 \neg (O(x_1, x_2, x_3, x_4) \land O(x_5, x_6, x_7, x_4) \land C(x_2, x_8) \land C(x_6, x_9) \land x_8 \neq x_9).
	\]
	The denial constraint $\mathrm{dc}_1$ essentially says that no customer should have more than one order for the same product and $\mathrm{dc}_2$ requires that no orders to different locations can be transported in the same shipment. 
	
	Both constraints are violated in the tables $C$ and $D$.
	The conflicts in the DB include $\{t_1,t_4\},\{t_1,t_3,s_1,s_2\}$.
	The respective SETAF $\calS_{\langle \{C, O\}, D \rangle}$ is depicted in Figure~\ref{fig:setaf_dcs}. 
	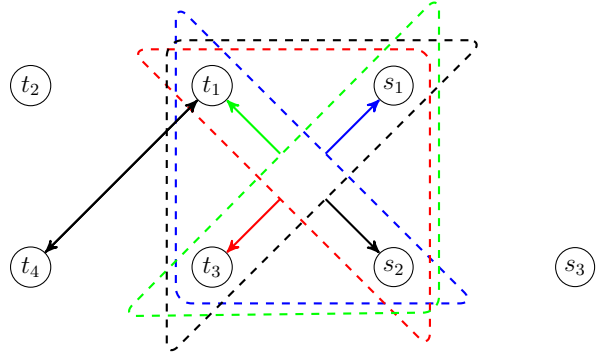
\begin{figure}
		\centering
		\begin{tikzpicture}[scale=1.2,arg/.style={circle,draw=black,fill=white, scale=0.8, inner sep=.75mm}]
			\node[arg] (t1) at (0,2) {$t_1$};
			\node[arg] (t3) at (0,0) {$t_3$};
			\node[arg] (t2) at (-2,2) {$t_2$};
			\node[arg] (t4) at (-2,0) {$t_4$};
			\node[arg] (s1) at (2,2) {$s_1$};
			\node[arg] (s2) at (2,0) {$s_2$};
			\node[arg] (s3) at (4,0) {$s_3$};

			\draw[rounded corners=6pt, dashed, thick,  draw=blue] 
			(-.4,2.9) -- (-.4,-.4) -- (2.9,-.4) -- cycle;
			\draw[rounded corners=6pt, dashed, thick, draw=red] 
			(2.4,2.4) -- (-.9,2.4) -- (2.4,-.9) -- cycle;
			\draw[rounded corners=6pt, thick, dashed, draw=green] 
			(2.5,3) -- (2.5,-.5) -- (-1,-.54) -- cycle;
			\draw[rounded corners=6pt, thick, dashed, draw=black] 
			(3,2.5) -- (-.5,2.5) -- (-.5,-1) -- cycle;
			
			\draw[-stealth', thick, draw=green] (.75, 1.25) -> (t1);
			\draw[-stealth', thick, draw=red] (0.75, 0.75) -> (t3);
			\draw[-stealth', thick, draw=blue] (1.25, 1.25) -> (s1);
			\draw[-stealth', thick, draw=black] (1.25, 0.75) -> (s2);
			\draw[-stealth', thick, draw=black] (t1) -> (t4);
			\draw[-stealth', thick, draw=black] (t4) -> (t1);
			
		\end{tikzpicture}
		
		\caption{SETAF $\calS_{\langle \{C, O\}, D \rangle}$ for Example~\ref{ex:setaf_dcs}. The attacker in each set-attack is depicted as a triangle of different color (for better presentation) and the attack is presented in the same color.  E.g., the red triangle and the arrow depicts the attack $(\{t_1,s_1,s_2\},t_3)$. }
		\label{fig:setaf_dcs}
	\end{figure}
\end{example}

	It is known that the repairs for $\calD$ correspond precisely to the independent sets of its conflict hypergraph $\calG_\calD$.
	Observe that the independent sets in $\calG_\calD$ are exactly the naive extensions of $\calS_\calD$.
	Moreover, using the observations from~\cite[Theorem 28]{BienvenuB20}, it is easy to see that every naive extension is also preferred as well as stable.
	Therefore, we have the following equivalence between the extensions and repairs.
	The equivalence regarding the stable and preferred semantics is a direct consequence of the results in~\cite{BienvenuB20}.
	\begin{theorem}\label{thm:ext-dcs}
		Let $\mathcal D = \langle \calT, D \rangle$ be a constrained database where $D$ is a set of DCs and let $\AF{S}{D}$ denote the SETAF generated by $\mathcal D$.
		Then for every $\calP\subseteq \calT$, $\calP \in \repairs(\calD)$ iff $\calP\in \sigma(\AF{S}{D})$ for $\sigma \in \{\naive, \stab, \pref\}$.
	\end{theorem}
	\begin{proof}
		Given a repair $\calP\subseteq \calT$, then $\calP\models d$ for every $d\in D$ and hence there is no $C\in\conflicts(\calD)$ with $C\subseteq \calP$. 
		So, the set of arguments in $\calP$ is conflict-free since every attack in $R$ has the form $(E\setminus\{t\},t)$ where
		$E\cup \{t\} \in \conflicts(\calD)$.
		
		The converse holds similarly.
		A naive extension $P\subseteq A$ does not contain all the arguments in any attack $(E,t)$. Thus, for every conflict $C\in \conflicts(\calD)$, such that $C=E\cup\{t\}$, we have that $C\not\subseteq P$ and therefore $P$ (seen as a subset of $\calT$) is a repair for $\calD$.
		Finally, the subset-maximality remains the same for both cases since each fact is seen as an argument and vice versa. 
		This establishes the correctness of our theorem.
	\end{proof}
	
	It is worth remarking that the presence of DCs with a single database atom yields singleton conflicts which results in self-contradictory facts (or self-attacking arguments in the translation).
	However, such self-contradictory facts do not occur in any repair and can be removed using existing reasoning algorithms, as highlighted by Bienvenu and Bourgaux~\cite{BienvenuB20}.
	Therefore, a usual assumption is that every denial constraint contains at least two database atoms.
	This has the effect that each conflict contains at least two argument.
	Consequently, in the presence of DCs, every fact in a conflict belongs to at least one repair and no fact in any conflict belongs to every repair.
	This holds because for each conflict $C\in\conflicts(\calD)$, there is a repair containing facts in $C\setminus \{t\}$ as well as another one containing $t$.
	As a result, we have the following observation regarding
	the acceptability of facts with respect to $\calD$.
	Specifically, we can decide if a given fact
	$s \in \calT$ is in some (or all) repairs, in polynomial time for a fixed set of DCs.
	
	\begin{remark}\label{cor:dc-cs}
		Let $D$ be a set of DCs over the schema $\tau$. 
		Then, for every constrained database $\calD = \langle \calT,D \rangle$ involving a database $\calT$ over $\tau$: $\somerepair$ is true for every $s\in\calT$ and $\allrepair$ is true iff $s\not\in C$ for any $C\in \conflicts({\calD})$.
		Moreover, both problems are decidable in polynomial time in data complexity.
	\end{remark}

	\subsection{Simulating LTGDs via SETAFs}\label{sec:tgds}
	 {Let $\calD =\langle \calT, L\rangle $ be a constrained database with a collection $\calT$ of facts and a collection $L$ of LTGDs.
		We recall the semantics for LTGDs via variable assignments.
		For a set $V$ of variables, an \emph{assignment} of $V$ into a database $\calT$ is a mapping $h\colon V\rightarrow \dom(\calT)$.
		For a relational atom $T(\tuple t)$ with terms $\tuple t=(t_1,\dots,t_k)$, we denote by $h(T(\tuple t))$ the result of substituting every $t_i$ in $T(\tuple t)$ by $h(t_i)$, where we set $h(t_i)=t_i$ when $t_i$ is a constant.
		{In other words, an assignment maps a relational atom $T(\tuple t)$ to a candidate fact $h(T(\tuple t))$, which may or may not be present in the database $\calT$.}
		The notation extends to a conjunction of atoms by considering them as a set. 
		{Precisely, for a conjunction $\phi(\mathbf x)\dfn\bigwedge_{1\leq i\leq n} T_i(\mathbf x_i)$ with $\mathbf x_i\subseteq \mathbf x$ and $T_i\in \tau$, and an assignment $h\colon \mathbf x\rightarrow \dom(\calT)$, we define $h(\phi(\mathbf x))=\{h(T_i(\mathbf x_i))\mid 1\leq i\leq n\}$.
		}
		Recall that every LTGD $\ell\in L$ has the form $\forall \tuple x (T(\tuple x) \rightarrow \exists \tuple y \psi(\tuple x,\tuple y))$ where $\psi\dfn \bigwedge_{1\leq i\leq n_\ell}\psi_i(\tuple z_i)$ 
		is a conjunction of 
		relational atoms over variables $\tuple z_i \subseteq \tuple x\cup\tuple y$ using relation names $\psi_i\in\tau$.
		Then, $\calT\models \ell$ if for each assignment $h$ such that $h(T(\tuple x))\in \calT$, there exists an assignment $h'\colon \mathbf x\cup\mathbf y\rightarrow \dom(\calT)$ such that $h'\upharpoonright_{\tuple x}=h\upharpoonright_{\tuple x}$ and $h'(\psi(\tuple x, \tuple y)) \subseteq \calT$, where $h\upharpoonright_{\tuple x}$ is the restriction of $h$ on $\tuple x$.
		If $h(T(\tuple x))\in \calT$ is true and there indeed exists such an assignment $h'$ 
		then we say that the set of facts $h'(\psi(\tuple x, \tuple y))$ \emph{supports} the fact $h(T(\tuple x))$ for the LTGD $\ell$.
		For such an LTGD $\ell$, we call $T$ the \emph{source} of $\ell$, denoted as $\source(\ell)$. For a $\source(\ell)$-fact $s$, we denote by $\support{\ell}{s}$ the collection of all the sets supporting $s$ for the LTGD $\ell$.}
	It is worth remarking that, for a constrained database $\calD=\langle\calT,L\rangle$ and $\ell\in L$, $\calD\models \ell$ if and only if $\support{\ell}{s}\neq\emptyset$ for each $\source(\ell)$-fact $s$.
	
	Interestingly, we can encode the sets supporting a fact by using SETAFs, as formalized next.
	
	\begin{definition}\label{def:setaf-tgd}
		Let $\calD =\langle \calT,L \rangle $ be a constrained database including a database $\calT$ and a collection $L$ of LTGDs.
		Then $\AF{S}{D}$ is the following SETAF.
		\begin{itemize}
			\item $A \dfn  \{s \mid s\in \calT \} \cup \{s_\ell \mid s\text{ is a } \source(\ell) \text{-fact for } \ell\in L\}$,
			\item $R\dfn \{(s_\ell,s), (s_\ell,s_\ell) \mid s\text{ is a } \source(\ell) \text{-fact for } \ell\in L \}\cup \{(S,s_\ell)\mid S\in \support{\ell}{s}, \ell\in L\}$.
		\end{itemize}
	\end{definition}
	 {Intuitively, for each $\ell\in L$ and {a $\source(\ell)$-fact $s$}, 
		the attack $(s_\ell,s)$ models that the LTGD $\ell\in L$ must be satisfied for the fact $s$.
		Moreover, the self-attacks $(s_\ell,s_\ell)$ enforce that the extensions only contain arguments corresponding to facts in the database.
		Then, each $S\in \support{\ell}{s}$ attacks $s_\ell$ and consequently, defends $s$ against $s_\ell$.
		The whole idea captured in this translation is that for each $\ell \in  L$ 
		of the form $\forall \tuple x (T(\tuple x) \rightarrow \exists \tuple y \bigwedge_{1\leq i\leq n_\ell}\psi_i(\tuple z_i))$ where $\tuple z_i \subseteq \tuple x\cup\tuple y$: 
		a $T$-fact $s$ is in a repair $\mathcal P$ of $\calD$ if and only if there are $\psi_i$-facts $t_i $ for $1\leq i\leq n_\ell$, such that $\{t_1,\dots,t_{n_\ell}\}\in \support{\ell}{s}$ if and only if the set $S=\{t_1,\dots,t_{n_\ell}\}$ attacks $s_\ell$ in $\calS_\calD$ if and only if the argument $s\in A$ is defended against $s_\ell$ in $\AF{S}{D}$.}

	\begin{example}
		\label{ex:setaf_lavtgds}
		The database tables $E$, $D$ and $P$ contain data about employees, departments and projects, respectively, and are defined as follows:
		\begin{center}
			\begin{tabular}{l@{\hskip 5pt}|@{\hskip5pt}c@{\;}cc@{\hskip5pt}}
				$E$ & Emp\_ID  & Dept\_ID\\ \hline
				$s_1$ & E1 & D1\\
				$s_2$ & E2 & D2\\
				$s_3$ & E3 & D3\\
			\end{tabular}
			\hspace{5mm}
			\begin{tabular}{l@{\hskip 5pt}|@{\hskip5pt}c@{\;}cc@{\hskip5pt}}
				$D$ & Dept\_ID  & Dept\_Name\\ \hline
				$t_1$ & D1 & Accounting\\
				$t_2$ & D3 & Sales\\
			\end{tabular}
			\hspace{5mm}
			\begin{tabular}{l@{\hskip 5pt}|@{\hskip5pt}c@{\;}cc@{\hskip5pt}}
				$P$ & Prod\_ID  & Dept\_ID\\ \hline
				$u_1$ & P1 & D1\\
				$u_2$ & P2 & D3\\
			\end{tabular}
		\end{center}
		Here, Emp\_ID, Dept\_ID and Prod\_ID are the IDs of individual employees, departments and products, respectively.
		Let $L = \{\mathrm{lav}_1, \mathrm{lav}_2\}$ be a set of LTGDs where
		\[
		\mathrm{lav}_1 \coloneqq \forall x_1, x_2 (E(x_1, x_2) \to \exists y, z (D(x_2, y) \land P(z, x_2)))
		\]
		and
		\[
		\mathrm{lav}_2 \coloneqq \forall x_1, x_2 (D(x_1, x_2) \to \exists y, z (E(y, x_1) \land P(z, x_1))).
		\]
		The constraint $\mathrm{lav}_1$ essentially says that if an employee is in a department, then the department must exist and be associated with a project.
		The constraint $\mathrm{lav}_2$ says that each department has at least one employee and at least one project associated with it.  
		
		In the tables $E$, $D$ and $P$, $\mathrm{lav}_1$ is violated while $\mathrm{lav}_2$ is satisfied.
		The source of each TGD $\ell$, and the supporting facts  for $\source(\ell)$-fact $s$ are presented as follows:
		\begin{itemize}
			\item $\source(\mathrm{lav}_1) = E,$ 
			\item $ \source(\mathrm{lav}_2) = D,$ 
			\item $\{t_1,u_1\}$ supports $s_1$ for $\mathrm{lav}_1$,
			\item $s_2$ does not have any supporter for $\mathrm{lav}_1$
			\item $\{t_2,u_2\}$ supports $s_3$ for $\mathrm{lav}_1$,
			\item $\{s_1, u_1\}$ supports $t_1$ for $\mathrm{lav}_2$,
			\item $\{s_3, u_2\}$ supports $t_2$ for $\mathrm{lav}_2$.
		\end{itemize} 
		The respective SETAF $\calS_{\langle \{E, D, P\}, L \rangle}$ is depicted in Figure~\ref{fig:setaf_lavtgds}.
		\begin{figure}
			\centering
			\begin{tikzpicture}[scale=.8,arg/.style={circle,draw=black,fill=white, scale=0.8, inner sep=.75mm}]
				\node[arg] (s1) at (0,2) {$s_1$};
				\node[arg] (s2) at (0,0) {$s_2$};
				\node[arg] (s3) at (0,-2) {$s_3$};
				
				\node[arg] (s1l) at (-2,2) {$s_{11}$};
				\node[arg] (s2l) at (-2,0) {$s_{12}$};
				\node[arg] (s3l) at (-2,-2) {$s_{13}$};
				\node[arg] (u1) at (2,2) {$u_1$};
				\node[arg] (u2) at (2,-2) {$u_2$};
				\node[arg] (t1) at (4,2) {$t_1$};
				\node[arg] (t2) at (4,-2) {$t_2$};
				
				\node[arg] (t1l) at (6,2) {$t_{21}$};
				\node[arg] (t2l) at (6,-2) {$t_{22}$};

				\draw[rounded corners=6pt, thick, green] (-.4, 2.4) rectangle (2.4, 1.6);
				\draw[rounded corners=6pt, thick, green] (-.4, -1.6) rectangle (2.4, -2.4);
				
				\draw[rounded corners=6pt, thick, blue] (1.6, 2.4) rectangle (4.4, 1.6);
				\draw[rounded corners=6pt, thick , blue] (1.6, -1.6) rectangle (4.4, -2.4);

				\foreach \f/\t in {s1l/s1l,s2l/s2l,s3l/s3l}{
					\path[-stealth',draw=red] (\f) edge[loop left] (\t);
				}
				\foreach \f/\t in {t1l/t1l,t2l/t2l}{
					\path[-stealth',draw=red] (\f) edge[loop right] (\t);
				}
				\foreach \f/\t in {s1l/s1,s2l/s2,s3l/s3,t1l/t1,t2l/t2}{
					\path[-stealth', thick, draw=black] (\f) edge (\t);
				}
				
				\path[->,draw=green, thick] (1,2.5) edge[bend left=60] (t1l);
				\path[->,draw=green, thick] (1,-1.5) edge[bend left=60] (t2l);
				
				\path[->,draw=blue, thick] (3,2.5) edge[bend right=60] (s1l);
				\path[->,draw=blue, thick] (3,-1.5) edge[bend right=60] (s3l);
				
				
			\end{tikzpicture}
			\caption{SETAF $\calS_{\langle \{E, D, P\}, L \rangle}$ for Example~\ref{ex:setaf_lavtgds}. For brevity, we rename the LTGDs $\{\mathrm{lav}_1,\mathrm{lav}_2\}$ to be $\{1,2\}$. 
				Moreover, the auxiliary arguments for $\source(1)$-facts $s_i$ are renamed to $s_{1i}$ and those for $\source(2)$-facts $t_j$ to $t_{2j}$.}
			\label{fig:setaf_lavtgds}
		\end{figure}
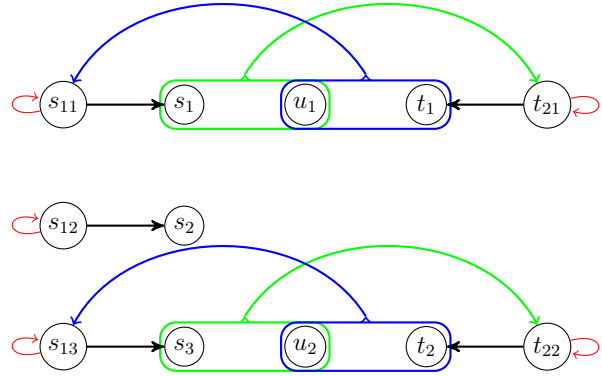
	\end{example}

		Similar to the case of DCs, the repair checking problem under LTGDs is $\DP$-complete if the set $L$ of ICs is considered as input~\cite[Table~3]{arming2016complexity}.
		Therefore, we consider a fixed set $L$ of LTGDs as the complexity drops to $\Ptime$ in this case.
		Then, for a constrained database $\calD=\langle\calT,L \rangle$, the framework $\calS_\calD$ can be constructed in polynomial time.
		An algorithm basically goes through every $\ell\in L$ in turn: for each $\source(\ell)$-fact $s$, it computes the set $\support{\ell}{s}$ in polynomial time since the size of each LTGD in $L$ is constant. 
		
		 {For a set $L$ of LTGDs, a minimal violation of $\ell\in L$ is caused by a single $\source(\ell)$-fact $s$ for which $\support{\ell}{s}=\emptyset$. 
			Since we only allow subset repairs (deleting facts), the only way to eliminate a violation is to remove such a fact $s$.
			However, this removal of facts may trigger other violations for facts, all of whose supporting sets contained some deleted facts.
			This triggers a chain reaction and the repair process iteratively removes any fact that has no remaining supporting set. Since the database is finite, this process eventually terminates (in the worst case there are no facts remaining and $\emptyset$ is the only repair). 
			The remaining facts are exactly those that have at least one full support left, so no violations remain.
			The key observation is that this process always leads to the same final set of facts, no matter in which order we remove unsupported facts. 
			Therefore, the final set of facts is consistent and subset-maximal.
			Observe that in some cases $\emptyset$ might be the only repair for a database $\calD$ involving LTGDs.
		}
		\begin{remark}\label{rem:unique}
			A constrained database $\mathcal D = \langle \calT, L \rangle$ involving a set $L$ of LTGDs admits a unique repair.
		\end{remark}
		 {Now we prove that the repair of an instance $\calD$ actually corresponds to the unique preferred extension of the resulting SETAF $\AF{S}{D}$.
		}	
		
		\begin{theorem}\label{thm:ext-tgds}
			Let $\mathcal D = \langle \calT, L \rangle$ be a constrained database where $L$ is a set of LTGDs and let $\AF{S}{D}$ denote the SETAF generated by $\mathcal D$.
			Then for $\calP\subseteq \calT$, $\calP \in \repairs(\calD)$  iff $\calP\in \pref(\AF{S}{D})$.
			Moreover, $\AF{S}{D}$ has exactly one preferred extension.
		\end{theorem}
		
		\begin{proof}
			We first prove the reverse direction.
			Let $P\subseteq A$ be a preferred extension in $\AF{S}{D}$, then $P$ must not contain an auxiliary argument $s_\ell$ corresponding to any LTGD $\ell\in L$ since $P$ is conflict-free.
			This implies that $P\subseteq \calT$, which together with the fact that $P$ is admissible (hence every $s\in P$ is defended against each $s_\ell\in A$) and maximal under set inclusion yields the proof of the claim.
			
			Conversely, let $\calP\subseteq \calT$ denote a repair for $\calD = \langle \calT, L \rangle$. 
			Then 
			$\calP$ is conflict-free in $\AF{S}{D}$ since each attack in $R$ contains at least one argument among the auxiliary arguments ($s_\ell$) which are not in $\calP$ (as $\calP\subseteq \calT$).
			Moreover, for each $s\in \calP$ and $\ell\in L$, there are $t_1,\dots,t_{n_\ell}\in \calP$, such that $\{s, t_1,\dots, t_{n_\ell}\}\models \ell$.
			This implies that each $s\in \calP$ is defended against the attack $s_\ell\in A$ since $(\{t_1,\dots,t_{n_\ell}\}, s_\ell)\in R$. 
			Consequently, $\calP$ is admissible. 
			To prove that $\calP$ is also preferred in $\AF{S}{D}$, assume to the contrary that there is an admissible $\calP'\supset \calP$ in $\AF{S}{D}$.
			Since $\calP'$ is also conflict-free, using the same argument as for $\calP$ we notice that $\calP'\subseteq \calT$.
			Now, $\calP'$ being admissible (together with the claim in reverse direction)
			implies that $\calP'$ is a repair for $\calD$ contradicting the fact that $\calP$ is a subset-maximal repair for $\calD$.
			As a consequence, $\calP$ is preferred in $\AF{S}{D}$.
			
			 {This proves the correctness of the theorem. The uniqueness of the extension follows due to Remark~\ref{rem:unique}.}
		\end{proof}
		
		Notice that a framework $\AF{S}{D}$ may not have stable extensions for certain constrained databases $\calD$ including databases $\calT$ and LTGDs $L$.
		This holds because some arguments can neither be accepted in an extension (e.g., when $\support{\ell}{s}=\emptyset$ for some $ \source(\ell)$-fact $s$ and $\ell\in L$), nor attacked by arguments in an extension (since arguments in $A$ only attack auxiliary arguments). 
		%
		Mahmood et al.~\cite{mahmood2024computing} presented a pre-processing over the argumentation framework generated by a constrained database $\calD$ involving inclusion dependencies, which computes its unique preferred extension which is also naive and stable.
		It is easy to observe that a similar pre-processing can be implemented for the case of LTGDs.
		This pre-processing also allows us to answer credulous and skeptical acceptance for a fact $s\in\calT$ in polynomial time considering the data complexity.
		
		\paragraph{A pre-processing algorithm for $\AF{S}{D}$.}
		Observe that an undefended argument in a SETAF $\calS=(A,R)$ cannot belong to any preferred extension of $\calS$.
		The intuition behind pre-processing is to remove such arguments, which are not defended against some of their attacks in $\AF{S}{D}$.
		This corresponds to (recursively) removing those facts $s\in \calT$, for which there is $\ell\in L$ such that $s$ is a $\source(\ell)$-fact and $\support{\ell}{s}=\emptyset$.
		The pre-processing (denoted $\pre(\AF{S}{D})$) applies the following procedure as long as possible. 
		\begin{description}
			\item[*] For each $s_\ell\in A$ such that $s_\ell$ is not attacked by any $S\subseteq A\setminus \{s_\ell\}$: remove $s$, $s_j$ for each $j\in L$, as well as each attack to and from $s$ and $s_j$.
		\end{description}
		We repeat this procedure until convergence. 
		Once a fixed point has been reached, the remaining arguments in $A$ are all defended. 
		Interestingly, after the pre-processing, removing the arguments with self-attacks results in a unique naive extension which is also stable and preferred.
		In fact, this naive extension corresponds to the unique repair of the instance $\calD$ (see Remark~\ref{rem:unique}).
		In the following, we also denote by $\pre(\AF{S}{D})$ the SETAF obtained after applying the pre-processing on $\AF{S}{D}$.
		Notice that $\pre$ is basically an adaptation to the SETAFs of the previously presented procedure for AFs and originates from the well-known algorithm for finding a maximal satisfying subteam for inclusion logic formulas~\cite[Lem.~12]{HannulaH22}.
		\begin{lemma}\label{lem:pre-setafs}
			Let $L$ be a fixed set of LTGDs, $\calD=\langle\calT, L\rangle$ be a constrained database, and $\AF{S}{D}$ denote the SETAF generated by $\calD$.
			Then $\pre(\AF{S}{D})$ can be computed from $\AF{S}{D}$ in polynomial time.
			Moreover, $\pre(\AF{S}{D})$ has a unique naive extension which is also stable and preferred.
		\end{lemma}
		
		\begin{proof}
			The procedure $\pre(\AF{S}{D})$ removes recursively all the arguments corresponding to facts $s$ such that $\support{\ell}{s}=\emptyset$ for some $\ell\in L$ and $s$ is a $\source(\ell)$-fact.
			Notice that $\support{\ell}{s}$ can be computed for each $\ell\in L$ and $ \source(\ell)$-fact $s$ in polynomial time for a fixed set $L$ of LTGDs.
			Then, $\pre$ stores in a data structure (such as a queue) all the arguments $s$ for which $\support{\ell}{s}=\emptyset$.
			Finally, each argument $s$ in this queue can be processed turn by turn, adding possibly new arguments when $\pre$ triggers the removal of certain arguments from $A$ and hence from $\support{\ell}{t}$ for some $t\in A$.
			A fixed-point is reached when every element in the queue has been processed, this gives the size of $A$ as the total number of iterations.
			Consequently, $\pre$ runs in polynomial time in the size of $\AF{S}{D}$.
			
			Let $\pre(\AF{S}{D}) = (A',R')$ denote the SETAF generated by the pre-processing.
			To prove the equivalence between extensions, notice that 
			the set of arguments $S$ without self-attacks in $A'$ forms a naive extension since $S$ is conflict-free and every argument in $A'\setminus S$ contains self-attacks. 
			Moreover, $S$  is also admissible since $\support{\ell}{s}\neq \emptyset $ for every argument $s\in S$ corresponding to a  $\source(\ell)$-fact $s$.
			Furthermore, since $A'\setminus S$ only includes auxiliary arguments, those are all attacked by $S$ and therefore $S$ is stable.
			Finally, $S$ is the only naive extension in the reduced AF since $S$ is the maximal conflict-free in $\pre(\AF{S}{D})$ and arguments in  $A'\setminus S$ contain self-attacks.
			
			This establishes the correctness of the lemma together with Theorem~\ref{thm:ext-tgds}.
		\end{proof}
		
		As a consequence of Lemma~\ref{lem:pre-setafs}, 
		we can also determine $\somerepair$ and $\allrepair$ for each $s\in \calT$, once the pre-processing has terminated resulting in $\pre(\AF{S}{D})$.
		
		\begin{remark}\label{cor:ltgd-cs}
			Let $\calD=\langle\calT,L\rangle$ be a constrained database involving a set $L$ of LTGDs and $\AF{S}{D}$ denote the SETAF generated by $\calD$.
			Then, $\somerepair$ and $\allrepair$ is true for every $s\in \calT$ iff $s\in \pre(\AF{S}{D})$.
			As a consequence, both problems are decidable in polynomial time in data complexity.
		\end{remark}
		
		
		\subsection{Simulating DCs and LTGDs via SETAFs}\label{sec:dcs-tgds}

		Consider an instance $\calD= \langle \calT,B \rangle$ with a database $\calT$ and a collection $B= D\cup L$ of DCs ($D$) and LTGDs ($L$).
		We first understand the role of applying pre-processing in the presence of both ICs.
		Removing those facts from $\calT$ failing some LTGD reduces the number of arguments in the resulting SETAF without affecting the connection between extensions and repairs.
		The removed facts are limited to those that cannot belong to any repair since we are in the setting of subset-repairs.
		However, one still cannot directly answer $\somerepair$ or $\allrepair$ (as done in Remark~\ref{cor:ltgd-cs}) for all facts $s\in\calT$ even after applying this pre-processing.
		Moreover, in the presence of both types of ICs, one may want to see the effect of the failure of LTGDs and DCs separately for certain facts. 
		In other words, one can distinguish a fact violating an LTGD since it does not have a supporting fact in $\calT$ from another fact violating an LTGD since all of its supporting facts also violate some DCs.
		In this case, the pre-processing allows to distinguish these facts even if they are not removed from the resulting SETAF.
		
		In the following, we do not apply pre-processing and instead consider all the arguments for facts in the database.
		The framework $\AF{S}{D}\dfn (A,R_{\text D}\cup R_{\text L})$ generated by $\calD$ is specified as below.

		\begin{definition}\label{def:setaf-dc-tgd}
			Let $\calD = \langle \calT, B\rangle$ be a constrained database including a database $\calT$ and a collection $B=D\cup L$ of DCs $D$ and LTGDs $L$. 
			Then, we construct the SETAF $\calS_\calD = (A,R_D\cup R_L)$ as follows.
			\begin{itemize}
				\item $A \dfn  \{s \mid s\in \calT \} \cup \{s_\ell \mid s\text{ is a } \source(\ell) \text{-fact for } \ell\in L\}$,
				\item $R_D\dfn \{(C\setminus\{t\}, t) \mid C\in\conflicts(\calD), t\in C\}$,
				\item $R_L\dfn \{(s_\ell,s), (s_\ell,s_\ell) \mid s\text{ is a } \source(\ell) \text{-fact for } \ell\in L \}\cup \{(S,s_\ell)\mid S\in \support{\ell}{s} \text{ for } \ell\in L\}$.
			\end{itemize}
			As before, we call $\calS_\calD$ the SETAF generated by $\calD$.
		\end{definition}

		
		 {A consequence of allowing both types of ICs (DCs and LTGDs) is that the preferred and naive extensions do not coincide in general. Moreover, both $\somerepair$ and $\allrepair$ are non-trivial and distinct (cf. Remarks~\ref{cor:dc-cs} and \ref{cor:ltgd-cs}).
			Furthermore, in the presence of both types of ICs, the repairs correspond to somewhat costly (that is, preferred) SETAF-semantics.}
		
		\begin{theorem}\label{thm:ext-both-setafs}
			Let $\mathcal D = \langle \calT, B \rangle$ be a constrained database where $B$ includes DCs and LTGDs.
			Further, let $\AF{S}{D}$ denote the SETAF generated by $\mathcal D$.
			Then for every subset $\calP\subseteq \calT$, $\calP \in \repairs(\mathcal D)$ iff $\calP \in \pref(\AF{S}{D})$.
		\end{theorem}
		\begin{proof}
			The correctness follows from the proofs of Theorem~\ref{thm:ext-dcs} and~\ref{thm:ext-tgds}.
			The conflict-freeness and admissibility of $\calP$ implies that each DC and LTGD in $B$, respectively, is true in $\calP$.
			The converse follows the same line of argument. 
			Finally, the maximality of repairs in $\calD$ corresponds to the maximality of extensions in the resulting SETAF $\AF{S}{D}$.
		\end{proof}

		Next we establish that allowing DCs and LTGDs renders the data complexity for repair existence ($\REP$), as well as credulous ($\somerepair$) and skeptical ($\allrepair$) reasoning for facts, same as, respectively,  
		the existence, credulous and skeptical reasoning for preferred semantics for AFs.
		These lower bounds are proven for the case of FDs and IDs (see Theorem~\ref{thm:cred-both-fixed} and~\ref{thm:skep-both}).
		Here we prove that the upper bounds hold even if we extend FDs to DCs and IDs to LTGDs, respectively.
		However, the set of ICs has to be fixed.

		\begin{theorem}\label{thm:cred-both-extended}
			The problems 
			$\REP$ and $\somerepair$ are both 
			\NP-complete in the data complexity for DCs and LTGDs.
		\end{theorem}
		\begin{proof}
			We prove that $\REPB$ and $\somerepairB$ are both in $\NP$ for any fixed set $B$ of DCs and LTGDs.
			Let $\mathcal D = \langle \calT, B \rangle$ be a constrained database involving a database $\calT$.
			For membership, one needs to guess a repair for $\calD$ that is non-empty, respectively contains $s$. 
			The verification can be done in $\Ptime$ since $B$ is a fixed set~\cite{arming2016complexity}.
			Notice that we do not need to check the maximality, since if there is a non-empty subset of $\calT$ satisfying every IC in $B$ (resp., a subset containing $s$) then there is also a non-empty repair for $\calD$ (containing $s$).
			
			Hardness follows from the case of FDs and IDs proved in Theorem~\ref{thm:cred-both-fixed}.
			To be precise, we construct a fixed set $B$ (resp., $B'$) of FDs and IDs for which 
			the problem $\REPB$ ($\somerepair_{\text B'}$) is $\NP$-hard.
		\end{proof}
		
		\begin{theorem}\label{thm:skep-both-extended}
			The data complexity of   $\allrepair$ is $\PiP$-complete for DCs and LTGDs.
		\end{theorem}
		\begin{proof}
			For membership, we prove that $\allrepairB$ is in $\PiP$ for any fixed set $B$ of DCs and LTGDs. 
			To this aim, we present the following $\SigmaP$ procedure for the complement of $\allrepairB$.
			Let $\mathcal D = \langle \calT, B \rangle$ be a constrained database involving a database $\calT$.
			One can guess a set $\calP$ of facts in $\calT$ satisfying each IC in $B$ as a counter-example for $s$, that is, $s\not\in \calP$ and $\calP$ is a (candidate) repair for $\calD$, which can be decided in $\Ptime$ since $B$ is fixed~\cite{arming2016complexity}.
			However, one has to use oracle calls to determine whether $\calP$ is a repair (hence maximal).
			 {Then, $\allrepair$ is true for $s$ if there exists no $\calP\subseteq \calT$ such that (i) $\calP$ satisfies each IC in $B$, (ii) $s\not\in \calP$, and (iii) $\calP$ is a repair.}
			This gives the stated upper bound of $\PiP$.
			
			For hardness, we prove that there exists a set $B$ of DCs and LTGDs such that the problem $\allrepairB$ is $\PiP$-hard.
			We prove this claim for the case of FDs and IDs in Theorem~\ref{thm:skep-both}.
		\end{proof}
		
		 {We conclude this section by noting that the combined complexity of $\REP$, $\somerepair$, and $\allrepair$ remains open for now, although the hardness transfer from the data complexity.}
		
		\section{The Case of Less Expressive ICs}\label{sec:rep-afs}
		
		In this section, we consider FDs and IDs, which are restricted classes of ICs inside DCs and LTGDs, respectively. 
		The two sub-cases are interesting due to the fact that one only requires binary conflicts/supports. 
		As a result, the classical Dung's AFs with binary attacks suffice for encoding repairs into extensions.
		
		Following the same theme as in Section~\ref{sec:dcs-tgds}, the first two subsections translate instances containing only one type of ICs to AFs. 
		Then, we combine both (functional and inclusion) dependencies in the third subsection.
		%
		As before, we adopt the construction for FDs from \cite{BienvenuB20}, but without fact priorities, thus resulting in the need for a weaker AF-semantics to capture repairs
		
		Having established that both FDs and IDs can be modeled in AFs via binary attacks, Section~\ref{sec:both} generalizes this approach by allowing both types of dependencies. 
		The primary result of this section is noteworthy: it demonstrates that the integrity constraints can be restricted from (1) DCs to FDs and (2) LTGDs to IDs, while the corresponding encoding into argumentation frameworks preserves the exact semantics, modulo the transition from SETAFs to AFs.
		
		
		\subsection{Simulating Functional Dependencies via AFs}\label{sec:fd}
		We transform an instance $\mathcal D = \langle \calT,D \rangle $ with database $\calT$ and a collection $D$ of FDs to an AF $\mathcal F_\calD$. 
		In the case of FDs, each conflict involves two facts failing some functional dependency $d\in D$.
		As a result, this conflict can be modeled in an AF by drawing a bi-directional attack between the two corresponding arguments.
		For this reason, we adhere to AFs rather than SETAFs. 
		Recall that each FD is defined over a single relation $T$ in the schema of $\calD$. 
		
		\begin{definition}\label{def:af-dep}
			Let $\calD =\langle \calT, D\rangle$ be a constrained database including a database $\calT$ and a collection $D$ of FDs.
			Then, $\mathcal F_{\calD}$ denotes the following AF.
			\begin{itemize}
				\item  $A\dfn \calT$, that is, each $s\in \calT$ is seen as an argument,
				\item $R\dfn \{(s,t),(t,s) \mid s,t \in \calT  \text{ and } \{s,t\}\not\models d  \text{ for some } d \in D\}$.
			\end{itemize}
			We call $\AF{F}{D} $ the argumentation framework generated by $\calD$ and call $R$ the \emph{conflict graph} for $\mathcal D$.
		\end{definition}
		
		Note that, for a given constrained database $\mathcal D$, the framework $\AF{F}{D}$ can be generated in polynomial time. 
		The attack relation $R$ is constructed for each $d\in D$ over $T$ in the schema of $\calD$ by taking each pair $s,t$ of $T$-facts in turn and checking whether $\{s,t\}\models d$ or not.
		Moreover, the set $D$ of FDs does not have to be fixed, as in the case of DCs.
		
		\begin{example}\label{intro:ex-dep}
			Consider $\calD = \langle T,D \rangle$ with database $T=\{s,t,u,v\}$ as depicted inside table in Figure~\ref{fig:ex-dep} and FDs $\{\depas{\text{Emp\_ID}}{\text{Dept}}$, $ \depas{\text{Sup\_ID}}{\text{Building}}\}$ over $T$. 
			Informally, each employee is 
			associated with a unique department and employees supervised by the same supervisor 
			work in the same building.
			Observe that, $\{s,t\}\not\models \depas{\text{Emp\_ID}}{\text{Dept}}$, $\{u,v\}\not\models \depas{\text{Emp\_ID}}{\text{Dept}}$, and $\{t,v\}\not\models \depas{\text{Sup\_ID}}{\text{Building}}$.
			The resulting AF $\AF{F}{D}$ is depicted on the right side  of  Figure~\ref{fig:ex-dep}.
			The preferred (as well as naive and stable) extensions of $\AF{F}{D}$ include $\{s,v\}, \{t,u\}$ and $\{s,u\}$. 
			Clearly, these three are the only repairs for $\calD$.
			%
			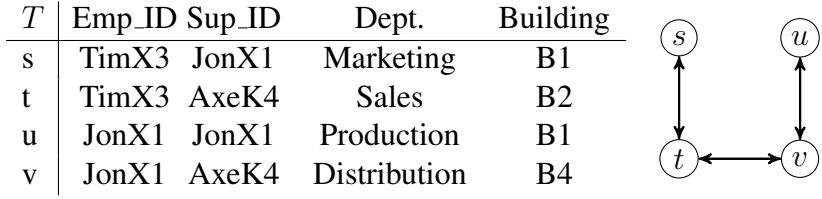
\begin{figure}[t]
				\centering
				\begin{tikzpicture}[scale=.8,arg/.style={circle,draw=black,fill=white,inner sep=.75mm}]
					\node[arg] (w) at (6,2) {$s$};
					\node[arg] (x) at (6,0) {$t$};
					\node[arg] (y) at (8,2) {$u$};
					\node[arg] (z) at (8,0) {$v$};
					
					\foreach \f/\t in {y/z,z/y,w/x,x/w,x/z,z/x}{
						\path[-stealth',draw=black, thick] (\f) edge (\t);
					}
					
					\node (table) at (0,1) {
						\begin{tabular}{l @{\hskip5pt}| @{\hskip 5pt}c@{\;}ccc}
							$T$ & Emp\_ID & Sup\_ID & Dept.  & Building \\
							\hline 
							s & TimX3 & JonX1 & Marketing	  & B1  \\
							t & TimX3 & AxeK4 & Sales 		  & B2 \\
							u & JonX1 & JonX1 & Production 	  & B1 \\
							v & JonX1 & AxeK4 & Distribution  & B4 \\
					\end{tabular}};
				\end{tikzpicture}\\[-.75em]
				\caption{Argumentation framework for modelling FDs in Example~\ref{intro:ex-dep}.}\label{fig:ex-dep}
			\end{figure}
		\end{example}
		
		It is easy to observe that any subset $\calP \subseteq \calT$ that satisfies each $d \in D$ contains precisely those facts in $\calT$ that are not pairwise conflicting.
		As a result, such subsets correspond to the naive extensions (maximal conflict-free sets) of $\AF{F}{D}$.
		Moreover, since the attack relation in $\AF{F}{D}$ is symmetric, i.e., $(s,t)\in R$  iff $(t,s) \in R$, the preferred, stable and naive extensions coincide~\cite[Prop. 4 \& 5]{Coste-MarquisDM05}.
		
		 {Observe that every FD is also a DC and hence the AF $\AF{F}{D}$ coincides with the SETAF $\AF{S}{F}$ for the constrained database $\mathcal D$.  As a consequence, we have the following corollary due to Theorem~\ref{thm:ext-dcs}. In particular, each conflict set $C$ (see proof of Theorem~\ref{thm:ext-dcs}) has size two for FDs.}
		\begin{corollary}\label{thm:ext-dep}
			Let $\mathcal D = \langle \calT, D \rangle$ be a constrained database where $D$ is a set of FDs and let $\AF{F}{D}$ denote the argumentation framework generated by $\mathcal D$.
			Then for every subset $\calP\subseteq \calT$, $\calP \in \repairs(\calD)$  iff $\calP\in \sigma(\AF{F}{D})$ for $\sigma \in \{\naive, \stab, \pref\}$.
		\end{corollary}

		Corollary~\ref{thm:ext-dep} reproves that a subset-repair for $\calD$ can be computed in polynomial time~\cite{DvorakDunne17}. 
		Moreover, similar to the case of DCs, we can decide if a given fact $s\in \calT$ is in some (or all) repairs, in polynomial time.
		In fact, it follows from the basic properties of functional dependencies that $\somerepair$ is true for  every $s\in \calT$, and $\allrepair$ is true for a fact $s\in \calT$ iff $\{s,t\} \models d$ for each $t\in \calT$ and $d\in D$.
		Furthermore, the latter problem can be decided in polynomial time for an input $\calD$ where the set $D$ of FDs does not have to be fixed (see Remark~\ref{cor:dc-cs}).

		We conclude this section by observing that 
		adding a size restriction on repairs renders the $\REP$ problem $\NP$-hard.
		Observe that the hardness follows due to a one-to-one correspondence between repairs and maximal independent sets of the conflict graph~\cite{LopatenkoB07}.
		However, we strengthen this result and note that the hardness already holds for propositional databases involving a single relation, that is, for a database $T$ with $\dom(T)=\{0,1\}$.
		The following result was proven in the context of team-semantics and maximal satisfying subteams for propositional dependence logic. 
		\begin{proposition}\cite[Theorem 3.32]{Mthesis23}\label{thm:size-dep}
			There is a constrained database $\calD$ including a propositional database $T$ and FDs $D$, such that given $k\in \mathbb N$, the problem to decide whether there is a repair $P\subseteq T$ for $\calD$ such that $|P|\geq k$ is $\NP$-complete.
		\end{proposition}
		
		\subsection{Simulating Inclusion Dependencies via AFs}\label{sec:inc}
		
		Let $\calD = \langle \calT,I\rangle $ be a constrained database with a database $\calT$ and collection $I$ of IDs.
		Similar to the case of LTGDs, for an ID $i\in I$ (say, $i\dfn  T[\tuple{x}]\subseteq T'[\tuple{y}]$), we call $T$ the \emph{source} of $i$. Moreover, we call $T'$ the \emph{target} of $i$.
		For $i \in  I$  and a $\source(i)$-fact $s$, let $t$ be $\target(i)$-fact such that $\{s,t\}\models i$.
		Then we say that $t$ \emph{supports} $s$ for the ID $i\in I$, and $\support{i}{s}$ 
		denotes the set of all such supporting facts for $s$ and an ID $i$. 
		Notice that for each $i\in I$ and $\source(i)$-fact $s$, we have that $\support{i}{s}\subseteq \target(i)$.
		Thus, the elements in $\support{i}{s}$ are facts instead of sets of facts as in the case of LTGDs. 
		Clearly, $\calD\models i$ if and only if $\support{i}{s}\neq\emptyset$ for each $\source(i)$-fact $s$.
		As for LTGDs, we create auxiliary argument $s_i$ for each $\source(i)$-fact $s$, which can be attacked by arguments corresponding to $\target(i)$-facts in $\support{i}{s}$.
		In the following, we formalize this notion and simulate the semantics for IDs via AFs.
		
		\begin{definition}\label{def:af-inc}
			Let $\calD =\langle \calT,I \rangle $ be a constrained database including a database $\calT$ and a collection $I$ of IDs.
			Then $\AF{F}{D}$ is the following AF.
			\begin{itemize}
				\item $A\dfn \calT \cup \{s_i \mid s \text{ is a }\source(i) \text{-fact for } i\in  I\}$,
				\item $R\dfn \{(s_i,s), (s_i,s_i) \mid s \text{ is a }\source(i) \text{-fact for } i\in I \} \cup \{(t,s_i)\mid t\in \support{i}{s} \text{ for } i\in  I\}$.
			\end{itemize}
		\end{definition}
		
		The intuition behind our encoding remains the same as for the case of LTGDs.
		That is, 
		the attack $(s_i,s)$ models that the ID $i\in I$ must be satisfied for the $\source(i)$-fact $s$ and the self-attacks $(s_i,s_i)$ enforce extensions to only contain arguments corresponding to facts.
		However, we only need attacks between arguments (instead of attacks from sets of arguments) as $\support{i}{s}\subseteq \target(i)$ for each $i\in I$ and $\source(i)$-fact $s$. 
		
		
		\begin{example}\label{intro:ex-inc}
			Consider $\calD = \langle T,I \rangle $ with database $T=\{s,t,u,v\}$ and IDs $I\dfn \{{\text{Sup\_ID}\subseteq \text{Emp\_ID}}, \\ {\text{Covers\_For}\subseteq \text{Dept}}\}$.
			For brevity, we denote IDs by $I= \{1,2\}$.
			The database and the supporting facts $\support{i}{w}$ for each $i\in I, w\in T$ are depicted in the table inside Figure~\ref{fig:ex-inc}.
			Informally, a supervisor is also an employee and each employee 
			is assigned a department to cover if that department is short on employees.
			For example, $s(\text{Sup\_ID})=t(\text{Emp\_ID})$, $s(\text{Covers\_For})=t(\text{Dept})=u(\text{Dept})$, and therefore $\support{1}{s}=\{t\}$, $\support{2}{s}=\{t,u\}$.
			Then we have the AF $\AF{F}{D}$ as depicted in Figure \ref{fig:ex-inc}.
			The AF $\AF{F}{D}$ has a unique preferred extension, given by $\{s,t\}$. 
			Clearly, this is also the only repair for $\calD$.
			\begin{figure}
				\centering
				\begin{tikzpicture}[scale=.8, arg/.style={circle,draw=black,fill=white,inner sep=.75mm}]
					\node[arg] (s) at (-.8,2) {$s$};
					\node[arg] (s1) at (-2,2) {$s_1$};
					\node[arg] (s2) at (-1,3) {$s_2$};
					\node[arg] (t) at (-.8,0) {$t$};
					\node[arg] (t1) at (-2,0) {$t_1$};
					\node[arg] (t2) at (-1,-1) {$t_2$};
					\node[arg] (u) at (1.7,1.8) {$u$};
					\node[arg] (u1) at (3,2) {$u_1$};
					\node[arg] (u2) at (2,3) {$u_2$};
					\node[arg] (v) at (1.5,0) {$v$};
					\node[arg] (v1) at (3,0) {$v_1$};
					\node[arg] (v2) at (2,-1) {$v_2$};
					\node (table) at (-1,5.5) {
						\begin{tabular}{l@{\hskip 5pt}|@{\hskip5pt}c@{\;}ccc@{\hskip5pt}}
							$T$ & Emp\_ID  & Sup\_ID & Dept. & Covers\_For\\ \hline
							s 	& JonX1  &  AxeK4 		&  Production 		& Marketing \\
							t 	& AxeK4  &  AxeK4 		&  Marketing  		& Production \\
							u 	& TimX3  &  JonX1 		&  Marketing  		& Distribution \\
							v 	& JonX1  &  AxeK4 		&  Distribution	 	& R\&D \\
					\end{tabular}};
					\node (table) at (6,5.5) {
						\begin{tabular}{cc}
							$S_{1}$ & $S_{2}$\\ \hline
							t & t,u   \\
							t & s  \\
							s,v & v  \\
							t & - \\
						\end{tabular}
					};
				
				\foreach \f/\t in {s1/s,s2/s,t1/t,t2/t,u1/u,u2/u, v1/v,v2/v}{
					\path[-stealth',draw=blue] (\f) edge (\t);
				}
				\foreach \f/\t in {u1/u1,v1/v1}{
					\path[-stealth',draw=red] (\f) edge[loop right] (\t);
				}
				\foreach \f/\t in {s1/s1,s2/s2,t1/t1,t2/t2,u2/u2,v2/v2}{
					\path[-stealth',draw=red] (\f) edge[loop left] (\t);
				}
				\foreach \f/\t in {t/s1,t/s2,u/s2, t/t1, s/t2, s/u1, v/u1, v/u2, t/v1}{
					\path[-stealth',draw=black] (\f) edge[bend left] (\t);
				}
				
			\end{tikzpicture}\\[-.75em]
			\caption{The AF $\AF{F}{I}$ modelling $\calI$ in Example~\ref{intro:ex-inc}: the red self-loops together with blue arcs depict the attacks for each fact $w\in T$ due to IDs $i\in I$ and the black arcs model the attacks due to the support set $\support{i}{w}$.}\label{fig:ex-inc}
		\end{figure}
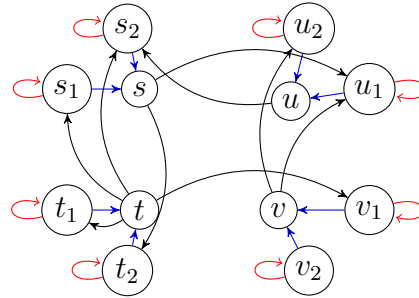
	\end{example}
	
	It is worth mentioning that $\{s,t,u,v\}$ constitutes a naive extension for $\AF{F}{D}$ in Example~\ref{intro:ex-inc}, although this is not a repair for $\mathcal D$.
	Clearly, the semantics for IDs in $\AF{F}{D}$ requires admissibility (defending against attacking arguments).
	
	 {Similar to the case of FDs, we observe that every ID is also an LTGD and hence the AF $\AF{F}{D}$ coincides with the SETAF $\AF{S}{D}$ for the constrained database $\mathcal D$. 
		As a consequence, the results for LTGDs transfer to IDs, modulo a translation from SETAFs to AFs.
		In particular, this allows us to apply Theorem~\ref{thm:ext-tgds} which establishes the following claim. 
	}

	\begin{corollary}\label{thm:ext-inc}
		Let $\mathcal D = \langle \calT, I \rangle$ be a constrained database where $I$ is a set of IDs and let $\AF{F}{D}$ denote the argumentation framework generated by $\mathcal D$.
		Then for $\calP\subseteq \calT$, $\calP \in \repairs(\mathcal D)$ iff $\calP\in \pref(\AF{F}{D})$.
		Moreover, $\AF{F}{D}$ has exactly one preferred extension.
	\end{corollary}
	The observation for LTGDs that some frameworks may not have any stable extension also applies to the special case of IDs and their corresponding AFs.
	As before, this holds because some arguments can neither be accepted in an extension, nor attacked by arguments in an extension. 
	The argument corresponding to the fact $v$ in Example~\ref{intro:ex-inc} depicts such an argument.
	Regarding the pre-processing, we still obtain the unique preferred extension which is also naive and stable.
	Importantly, the set $I$ of IDs does not have to be fixed to require the polynomial time pre-processing.
	Following Remark~\ref{cor:ltgd-cs}, $\somerepair$ and $\allrepair$ is true for every $s\in \calT$ such that $s\in \pre(\AF{F}{D})$.

	\begin{example}[Continued]
		Reconsider the constrained database $\calD= \langle T, I\rangle $ from Example~\ref{intro:ex-inc}.
		Observe that the argument $v$ is not defended against $v_2$ and therefore cannot be in a repair. 
		The pre-processing removes $\{v,v_1,v_2\}$ and all the edges to/from arguments in this set.
		This has the consequence that all the arguments which are only defended by $v$ are no longer defended (e.g., $u$).
		Consequently, the arguments $\{u,u_1,u_2\}$ have to be removed as well. 
		After repeating the same process for $u$, we notice that no further argument needs to be removed. 
		Hence, the set $\{s,t\}$ yields a repair for $\calD$ as well as a $\sigma$-extension in the reduced AF for $\sigma\in\{\naive,\stab,\pref\}$.
	\end{example}
	
	Our construction of AFs for IDs relies on self-attacking arguments.
	However, these self-attacks are restricted to only auxiliary arguments and allow us to avoid taking these arguments in any extension.
	This results in obtaining a precise connection between repairs of a DB and extensions of the resulting AF.
	Without self-attacking arguments, one gets extensions which may contain auxiliary arguments and thus a precise connection between repairs and extensions might be lost.
	This is further highlighted in the following example.
	
	\begin{example}\label{without-selfloops}
		Consider a single binary relation $T(x,y)$ in the schema, a constrained database $\calD = \langle T, I\rangle $ with database $T=\{s, t, u, v\}$ and a single ID $T[x] \subseteq T[y]$.
		The database table as well as its resulting AF $\AF{F}{D}$ (without self-attacks for auxiliary arguments) is depicted in Figure~\ref{fig:without-selfloops}. 
		Observe that $\AF{F}{D}$ admits the following four preferred extensions:
		$\{\{s,t, u_1, v\}, \{s,t, u_1, v_1\}, \{s_1,t_1, u_1, v\}, \{s_1,t_1, u_1, v_1\}\}$.
		However, there is a unique repair for $\calD$ given as $\{s, t, v\}$.
		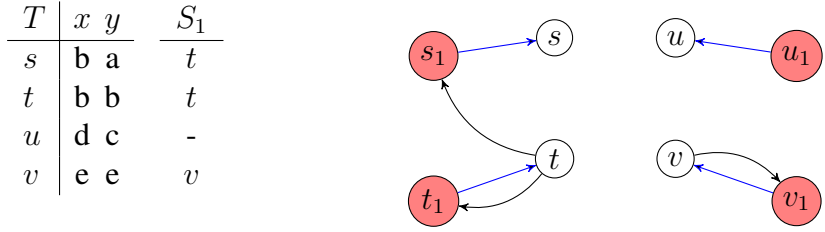
\begin{figure}[t]
			\centering
			\begin{tikzpicture}[scale=.8,arg/.style={circle,draw=black,fill=white,inner sep=.75mm}]
				\node[arg] (s) at (6,2) {$s$};
				\node[arg] (t) at (6,0) {$t$};
				\node[arg] (u) at (8,2) {$u$};
				\node[arg] (v) at (8,0) {$v$};
				
				\node[arg,fill = red!50] (s1) at (4,1.7) {$s_1$};
				\node[arg, fill = red!50] (t1) at (4,-.7) {$t_1$};
				\node[arg, fill = red!50] (u1) at (10,1.7) {$u_1$};
				\node[arg, fill = red!50] (v1) at (10,-.7) {$v_1$};
				
				\foreach \f/\t in {s1/s,t1/t,u1/u, v1/v}{
					\path[-stealth',draw=blue] (\f) edge (\t);
				}
				\foreach \f/\t in {t/s1,t/t1, v/v1}{
					\path[-stealth',draw=black] (\f) edge[bend left] (\t);
				}

				\node (table1) at (-2,1) {
					\begin{tabular}{l@{\hskip 5pt}|@{\hskip5pt}c@{\hskip5pt}c@{\hskip5pt}c@{\hskip5pt}}
						$T$ & $x$  & $y$ \\ \hline
						$s$ & b & a\\
						$t$ & b & b\\
						$u$ & d & c\\
						$v$ & e & e\\
				\end{tabular}};
				\node (table2) at (0,1) {
					\begin{tabular}{c}
						$S_{1}$ \\ \hline
						$t$  \\
						$t$  \\
						- \\
						$v$  \\
					\end{tabular}
				};
				
			\end{tikzpicture}
			\caption{Argumentation framework without self-attacks for modelling IDs in Example~\ref{without-selfloops}. The auxiliary arguments are highlighted in red for convenience.}\label{fig:without-selfloops}
		\end{figure}
	\end{example}
	
	Towards the end of our paper (Section~\ref{sec:self-attacks}), we will discuss how and whether we can model IDs via another translation that avoids self-attacking arguments. 
	
	\subsection{Simulating Functional and Inclusion Dependencies via AFs}\label{sec:both}

	Consider an instance $\calD = \langle \calT, B\rangle $ of a database $\calT$ and a collection $B = D\cup I$ of FDs $D$ and IDs $I$.
	We construct the AF~$\AF{F}{D}=(A,R_\text{D}\cup R_\text{I})$ as follows.
	\begin{itemize}
		\item $A \dfn  \{s \mid s\in \calT \} \cup \{s_i \mid s \text{ is a }\source(i) \text{-fact for } i\in I\}$,
		\item $R_\text{D} \dfn \{(s,t),(t,s) \mid s,t \in \calT  \text{ and } \{s,t\}\not\models d  \text{ for some } d \in D\}$,
		\item $R_{\text{I}}\dfn \{(s_i,s), (s_i,s_i) \mid s \text{ is a }\source(i) \text{-fact for } i\in I \}\cup \{(t,s_i)\mid t\in \support{i}{s} \text{ for } i\in I\}$.
	\end{itemize}
	
	As before, one may apply the pre-processing as a first step, thereby removing those facts from $\calT$ failing some $i\in I$.
	However, the discussion as for the case of LTGDs still applies for IDs.
	The following example illustrates that even if we apply pre-processing, some facts may not be accepted in combination with each other due to the presence of FDs.
	
	\begin{example}\label{intro:ex-both}
		Consider $\calD = \langle T,B \rangle$ with database $T=\{s,t,u\}$ and ICs $B= D\cup I$ where $D= \{\depas{\text{Sup\_ID}}{\text{Building}}\}$ and $I=\{{\text{Covers\_For}\subseteq \text{Dept}}\}$.
		Moreover, the database $T$ and the support $\support{\incas{\text{Covers\_For}}{\text{Dept}}}{w}$ for each $w\in T$ is depicted in the table inside Figure~\ref{fig:ex-both}.
		Then,  $\{s,t\}\not\models \depas{\text{Sup\_ID}}{\text{Building}}$, and $\{t,u\}\not\models \depas{\text{Sup\_ID}}{\text{Building}}$.
		The resulting AF $\AF{F}{D}$ is shown in Figure \ref{fig:ex-both}, where the edges due to the IDs are depicted in red and blue. 
		The only preferred extension for $\AF{F}{D}$ is $\{t\}$.
		Also, the only repair for $\calD$ is $\{t\}$.
		Further, although $\{s,u\}$ is preferred for $\AF{F}{D'}$ where $\calD' =\langle T,D\rangle$ (ignoring red and blue arcs), and $\{s,t, u\}$ is preferred for $\AF{F}{D''}$ where $\calD'' =\langle T,I\rangle$ (ignoring black arcs), none of them is preferred for $\AF{F}{D}$.
		\begin{figure}
			\centering
			\begin{tikzpicture}[scale=.8,arg/.style={circle,draw=black,fill=white,inner sep=.75mm}]
				\node[arg] (s) at (-5,2) {$s$};
				\node[arg] (s1) at (-5,0) {$s_1$};			
				\node[arg] (t) at (-3,2) {$t$};
				\node[arg] (t1) at (-3,0) {$t_1$};
				\node[arg] (u) at (-1,2) {$u$};
				\node[arg] (u1) at (-1,0) {$u_1$};
				\node (table) at (-3.5,3.5) {
					\begin{tabular}{l@{\hskip 5pt}|@{\hskip5pt}c@{\;}cccc}
						$T$ & Emp\_ID  & Sup\_ID & Dept. & Building & Covers\_for \\
						\hline 
						s & JonX1 & AxeK4 & Production	& B4  & Sales \\
						t & TimX3 & AxeK4 & Sales  			& B2  & Sales \\
						u & AxeK4 & AxeK4 & Marketing  		& B4  & Production\\
				\end{tabular}};
			\node (table) at (3.5,3.5) {
				\begin{tabular}{c}
					$S_i$ \\ \hline
					t \\
					t \\
					s\\
				\end{tabular}
			};
			
			\foreach \f/\t in {s/t,t/s,t/u,u/t}{
				\path[-stealth',draw=black] (\f) edge (\t);
			}
			\foreach \f/\t in {t/s1,t/t1,s/u1,s1/s,t1/t,u1/u}{
				\path[-stealth', draw=blue] (\f) edge (\t);
			}
			\foreach \f/\t in {s1/s1,t1/t1,u1/u1}{
				\path[-stealth',draw=red] (\f) edge[loop right] (\t);
			}
		\end{tikzpicture}\\[-.75em]
		\caption{Argumentation framework for modelling dependencies in Example~\ref{intro:ex-both}. Black arcs depict conflicts due to functional, and blue ones due to inclusion dependency.}\label{fig:ex-both}
	\end{figure}
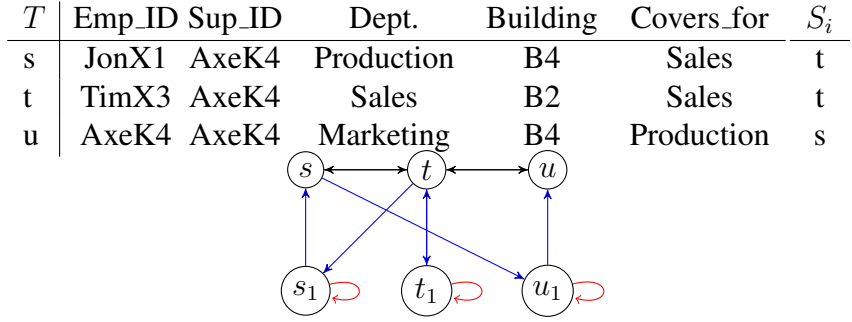
\end{example}

Similar to their expressive counterparts, allowing both types of ICs (FDs and IDs) has the effect that the preferred and naive extensions do not coincide in general. 
Moreover, the repairs correspond to preferred semantics for AFs.

\begin{example}[Cont.]
	Reconsider the constrained database $\calD$ from Example~\ref{intro:ex-both}.
	Then, $\{s,u\}$ is a naive extension for $\AF{F}{D}$ but not preferred.
	Moreover, $t$ is the only fact for which $\somerepair$ and $\allrepair$ is true. 
\end{example}


\begin{corollary}\label{thm:ext-both}
	Let $\mathcal D = \langle \calT, B \rangle$ be a constrained database where $B$ includes FDs and IDs.
	Further, let $\AF{F}{D}$ denote the argumentation framework generated by $\mathcal D$.
	Then for every subset $\calP\subseteq \calT$, $\calP \in \repairs(\mathcal D)$ iff $\calP \in \pref(\AF{F}{D})$.
\end{corollary}
%

Next we establish the data and combined complexity results for $\REP, \somerepair$ and $\allrepair$.
Observe that a fixed collection of integrity constraints suffice for establishing the hardness results (see Table \ref{table:cont}).
Whereas, the membership in each case also applies to the case when the set of constraints is given as input (hence not fixed).

\begin{theorem}\label{thm:cred-both-fixed}
	The problems $\REP$ and $\somerepair$ are in $\NP$ for constraints that consist of FDs and IDs.
	Moreover, there exists a set $B$ of FDs and IDs such that the problems 
	$\REPB$ and $\somerepairB$ are both 
	\NP-hard.
\end{theorem}
\begin{proof}
	The membership follows analogous to the proof of Theorem~\ref{thm:cred-both-extended} in both cases.
	Let $\mathcal D = \langle \calT, B \rangle$ be a constrained database where $B$ consists of FDs and IDs. Further, let $s\in \calT$ be a fact.
	Observe that for constraints involving FDs and IDs, the verification that a guessed set of facts satisfies  $B$ can be performed in polynomial time given $\calD$ as input.
	Therefore, the $\NP$-membership applies to $\REP$ and $\somerepair$.

	For hardness, we reduce from $\SAT$ to $\somerepairB$ and $\REPB$ simultaneously for a fixed set $B$ of constraints.
	Let $\varphi \dfn \{c_i\mid {i\leq m}\}$ be a propositional formula in conjunctive normal form over propositions $X= \{p_1,\ldots, p_n\}$.
	We construct a database with two relations\footnote{We find it convenient to write explicit attribute names for each relation.}: 
	 {(i) a $4$-ary relation $F(t_0,u_0,t_1,u_1)$, and (ii) a binary relation $C(t_2,u_2)$.}
	Intuitively, $F$ encodes the positive and negative participation of propositions in the clauses of $\varphi$. 
	Precisely, a fact $F(x,v,c_i,\text{sat})$ encodes that the proposition $x$ appears positively in clause $c_i$ when $v=1$, and negatively when $v=0$. The value ``sat'' appears in each fact to enforce a non-empty set of facts in a repair.
	The relation $C$ is used to encode the set of clauses. 
	With the help of IDs, we aim to enforce that if any $C$-fact belong to a repair, then all of the $C$-facts (a ``chain'') must belong to such a repair, which in turn forces (via further IDs) a repair of $\calT$ to include $F$-facts in such a way that their corresponding literals satisfy every clause.
	
	Formally, we will define $F$ and $G$ such that $\dom(F)\subseteq\{c_i \mid i\leq m\}\cup\{p_1,\ldots, p_n\}\cup\{0,1,\text{sat}\}$ and $\dom(C)=\{c_i \mid i\leq m\}\cup\{\text{sat}\}$.
	Our database consists of the following facts:
	%
	\begin{align*}
		\calT  = \quad & \{F(x,1,c_i,\text{sat}) \mid x \in c_i, c_i\in\varphi\}\cup \{F(x,0,c_i,\text{sat}) \mid \neg x \in c_i, c_i\in\varphi\} \\
		& \cup  \{F(\text{sat}, \text{sat}, \text{sat}, \text{sat})\} \cup 
		\{C(\text{sat}, c_1), C(c_m, \text{sat})\} \\
		& \cup \{C(c_i, c_{i+1}) \mid 1\leq i< m\}.
	\end{align*}
	%
	
	As ICs, we consider a collection $B$ of FDs and IDs over the schema $\{F,C\}$ using $V \dfn \{t_j,u_j \mid 0\leq j \leq 2\}$ as the set of  {attributes}\footnote{Recall that our notation for FDs and IDs uses attributes as names for the corresponding argument positions of the relation (here $F(t_0,u_0,t_1,u_1)$ and $C(t_2,u_2)$).}.
	Precisely, we let 
	$$B= \{\depas{t_0}{u_0}\} \cup \{F[u_1]\subseteq C[u_2], C[u_2]\subseteq C[t_2], C[t_2]\subseteq F[t_1], F[u_1]\subseteq F[t_1]\}.$$

	The FD $\depas{t_0}{u_0}$ over $F$ ensures that each proposition (corresponding to facts in a repair) takes at most one value in $\{0,1\}$.
	The IDs in $B$ ensure that each clause $c_i\in\varphi$ is satisfied by modeling the following intuition.
	(I) The ID $F[u_1]\subseteq C[u_2]$ together with a dummy fact $F(\text{sat}, \text{sat}, \text{sat}, \text{sat})$ (referred to as $s_d$) encodes that any repair for $\calT$ containing $s_d$ triggers a non-empty subset of $C$-facts in a repair. 
	(II) The ID $C[u_2]\subseteq C[t_2]$ triggers a chain of $C$-facts in a repair and thus enforces that either $\emptyset$, or all $C$-facts are considered in a repair.
	(III) The ID $C[t_2]\subseteq F[t_1]$ encodes that all clauses must be satisfied if a repair contains any $C$-fact. 
	Moreover, (IV) the ID $F[u_1]\subseteq F[t_1]$ is used in the reduction for $\REPB$ to enforce that a valid repair must contain the dummy fact $s_d$.
	Notice that the collection $B$ does not depend on the formula $\varphi$, hence it is fixed. 

	We next highlight some observations before delving into the correctness proof.
	Clearly, $\calD$ is inconsistent due to the presence of at least one pair of facts $F(x,1,c_i,\text{sat}), F(x,0,c_j,\text{sat})$ and the FD $\depas{t_0}{u_0}$ (based on the fact that at least one variable $x$ appears positively in one clause $c_i$ and negatively in another clause $c_j$, since otherwise $\varphi$ is trivially satisfiable).
	Then, a set $P_X$ of $F$-facts satisfies $\depas{t_0}{u_0}$ if for each proposition $x\in X$: $P_X$ contains facts of the form either $F(x,1,*,\text{sat})$ or $F(x,0,*,\text{sat})$ but not both (here $*$ indicates the occurrence of any value in a fact). 
	For $P_X$ to be a candidate repair among $F$-facts, it additionally has to contain facts of the form $F(x,*,*,*)$ for each $x\in X$.
	Moreover, a set $P_c$ of $C$-facts satisfies the ID $C[u_2]\subseteq C[t_2]$ iff either $P_c=\emptyset$ or $P_c$ contains all $C$-facts.
	Next, we consider a set $P_X$ of $ F$-facts containing the fact $s_d$ and let $\calP= P_X \cup P_c$. 
	It follows that $\calP\models F[u_1]\subseteq C[u_2]$ iff $P_c\neq \emptyset$ 
	and hence we let $P_c$ contain all $C$-facts.
	Finally, $\calP\models C[t_2]\subseteq F[t_1]$ iff for each $c_i\in\varphi$, $P_X$ contains at least one fact $F(x,v,c_i,\text{sat})$ 
	such that $v=1$ if $x\in c_i$ and $v=0$ if $\neg x\in c_i$.
	This is guaranteed since $\calP$ contains facts of the form $C(*, c )$ for each clause $c\in\varphi$ and thus there must exists $F$-facts in $\calP$ also of the form $F(*,*,c,*)$ for each clause $c\in \varphi$.
	We prove the correctness via the following claim.
	\begin{claim}
		$\varphi$ is satisfiable if and only if $\somerepair$ is true for the fact $s_d$.
	\end{claim}
	\textbf{Proof of Claim}
	``$\Longrightarrow$''. Suppose $\varphi$ is satisfiable and let $\theta$ be a satisfying assignment.
	We let $\calP \dfn \{F(x,1,c_i,\text{sat}) \mid \theta(x)=1, \text{ for } x\in X \text{ and } i\leq m\} \cup \{F(x,0,c_i,\text{sat}) \mid \theta(x)=0, \text{ for } x\in X \text{ and } i\leq m\} \cup \{s_d\} \cup P_c$ where $P_c$ contains all $C$-facts in $\calT$.
	Then $\calP\models \depas{t_0}{u_0}$
	because for any $x\in X$: $\calP$ contains either facts of the form $F(x,1,*,*)$ or $F(x,0,*,*)$, since $\theta$ is an assignment over $X$.
	Next, we have that 
	(I) $\calP\models F[u_1]\subseteq C[u_2]$ since $P_c \neq \emptyset$, and 
	(II) $\calP\models C[u_2]\subseteq C[t_2]$ since $\calP$ contains all $C$-facts.
	Finally, $\calP$ contains at least one fact of the form $F(*,*,c,*)$ for each clause $c\in \varphi$ since $\theta(\ell)=1$ for at least one literal $\ell\in c$ in each $c\in \varphi$.
	This holds because for every clause $c\in \varphi$, $\calP$ contains a fact $F(x,v, c, \text{sat})$ for some variable $x$ appearing in $c$.
	The value $v$ in this fact depends on how the assignment $\theta$ satisfies $c$, that is, $v=1$ if $\theta(x)=1$ makes $c$ true, and $v=0$ if $\theta(x)=0$ does.
	As a consequence, $\calP\models C[t_2]\subseteq F[t_1]$. 
	Observe that $\calP\models F[u_1]\subseteq F[t_1]$ follows due to the fact $s_d\in \calP$.
	To see why $\calP$ is a repair (i.e., subset-maximal), we note that the only facts not in $\calP$ are of the form $F(x,v,c,\text{sat})$ for some $v\in\{0,1\}$ and $c\in\varphi$ such that $F(x,1-v,c,\text{sat})\in \calP$ due to the way $\calP$ is defined.
	However, adding such a fact would violate the FD. Therefore, $\calP$ is indeed a repair.
	We conclude by observing that $s_d \in \calP$.
	
	``$\Longleftarrow$''.
	Let $\calP$ be a repair and $s_d \in \calP$.
	From $\calP\models F[u_1]\subseteq C[u_2]$, it follows that $\calP$ must contain at least one $C$-fact as $s_d\in\calP$, and from $\calP\models F[u_1]\subseteq C[u_2]$, it follows that $\calP$ actually contains all $C$-facts.
	We define an assignment $\theta$ over $X$ by setting $\theta(x)= 1$ if $\calP$ contains $F(x,1,*,*)$ and $\theta(x)= 0$ if it contains $F(x,0,*,*)$.
	Since $\calP\models \depas{t_0}{u_0}$, the assignment $\theta$ is well-defined, i.e., there is no $x\in X$ for which we assign both $\theta(x)=1$ and $\theta(x)=0$.
	Now, $\theta$ is defined for all the propositions in $X$ since $\calP$ is subset-maximal (hence, it contains fact of the form either $F(x,0,*,*)$ or $F(x,1,*,*)$ for each $x\in X$).
	Now, we prove that $\theta\models\varphi$.
	To this aim,  let $c\in\varphi$ be any clause.
	Since $\calP\models C[t_2]\subseteq F[t_1]$, $\calP$ contains at least one fact of the form $F(x,v,c,\text{sat})$ for some variable $x$ and $v\in\{0,1\}$, such that $v=1$ if $x\in c$ and $v=0$ if $\neg x\in c$ (by definition of how $F$-facts are constructed).
	Therefore, we must have set $\theta(x)=1$ if $v=1$ (hence $x\in c$) and $\theta(x)=0$ if $v=0$ ($\neg x\in c$). 
	As a result, $\theta\models c$ for each $c\in\varphi$.
	This proves $\theta\models \varphi$ and our claim follows. $\hfill\blacksquare$
	
	To reduce $\SAT$ into $\REPB$, we observe that $\calP\models F[u_1]\subseteq F[t_1]$ iff either $\calP$ contains no $F$-fact, or it contains the fact $s_d$.
	However, there cannot be a repair for $\calD$ without any $F$-fact but still containing $C$-facts, since it violates the ID $C[t_2]\subseteq F[t_1]$.
	As a result, one cannot construct a non-empty repair for $\calD$ by excluding $s_d$. 
	Therefore, every repair $\calP$ for $\calD$ necessarily contains $s_d$, thereby proving the equivalence as before.
	In other words, there is a non-empty repair for $\calD$ iff $\varphi$ is satisfiable.
	
	This completes the proof for both cases.
\end{proof}

We provide an example for better understanding of the reductions from the proof of Theorem~\ref{thm:cred-both-fixed}.
\begin{example}[Example with Fixed ICs]\label{ex:red-both-fixed}
	Let $\varphi \dfn \{x \lor y, \neg x \lor \neg y, \neg x \lor y\}$ be a propositional formula.
	Our reduction for $\somerepair$ gives a database $\calT$ with two tables $F$ and $C $ for the collection $B\dfn \{\depas{t_0}{u_0}\}\cup \{F[u_1]\subseteq C[u_2], C[u_2]\subseteq C[t_2], C[t_2]\subseteq F[t_1]\}$ of dependencies. 
	
	Observe that the only satisfying assignment for $\varphi$ is given by $\{x\mapsto 0, y\mapsto 1\}$, which corresponds to the  repair $\{S, \bar x_2, \bar x_3, y_1, y_3\}\cup \{s_c,s_1,s_2,s_3\}$ for $\calD$ containing $S$. Consequently $\somerepair$ is true for the fact $s_d$.
	\begin{table}[t]
		\centering
		\begin{tabular}{l@{\hskip5pt} |@{\hskip5pt} cc @{\hskip5pt}| @{\hskip5pt}c @{\hskip5pt}| @{\hskip5pt} c}
			$F$ & $t_0$ & $u_0$  & $t_1$ & $u_1$\\ \hline
			$s_d$ & sat & sat  & sat & sat\\ \hline
			$x_1$ & $x$ & 1  & $c_1$ & sat\\
			${\bar x_2}$ & $x$ & 0 & $c_2$ & sat \\
			${\bar x_3}$ & $x$ & 0 & $c_3$ & sat \\ \hline 
			$y_1$ & $y$ & 1  & $c_1$ & sat \\
			${\bar y_2}$ & $y$ & 0 & $c_2$ & sat \\
			$y_3$ & $y$ & 1  & $c_3$ & sat \\
		\end{tabular}
		\hspace{1cm}
		\begin{tabular}{l@{\hskip5pt} |@{\hskip5pt} cc}
			$C$& $t_2$  & $u_2$  \\\hline
			$s_c$ & sat  & $c_1$ \\
			$s_1$ & $c_1$  & $c_2$ \\
			$s_2$ & $c_2$  & $c_3$ \\ 
			$s_3$ & $c_3$  & sat \\
		\end{tabular}
		\caption{The database corresponding to the formula $\varphi$ from Example~\ref{ex:red-both-fixed}.}
		\label{tab:red-both-fixed}
	\end{table}
\end{example}

Next, we prove that $\allrepair$ is even harder and $\PiP$-complete.

\begin{theorem}\label{thm:skep-both}
	The problem $\allrepair$ is contained in $\PiP$ for FDs and IDs.
	Moreover, there exists a set $B$ of FDs and IDs such that the problem 
	$\allrepairB$ is 
	$\PiP$-hard.
\end{theorem}
\begin{proof}
	The membership follows analogous to the proof of  Theorem~\ref{thm:skep-both-extended}.
	The only difference here is that one can check in polynomial time whether a set of facts satisfies each IC in the input.
	As a result, the membership holds for $\allrepair$ (i.e., in the combined complexity).
	
	For hardness, we build on the same idea as in the proof of Theorem~\ref{thm:cred-both-fixed}.
	Here,  we reduce from an instance $\Phi$ of the $\PiP$-complete problem $\tqbf$.
	The problem $\tqbf$ is defined as follows: given a propositional formula $\Phi = \forall Y \exists Z \varphi(Y,Z)$ where $\varphi \dfn \{c_i \mid 1\leq i\leq m\}$ is a $\CNF$ formula, determine whether $\Phi$ is true (i.e., whether for every assignment to the variables in $Y$, there exists an assignment to the variables in $Z$ that satisfies $\varphi$).
	We let $X=Y\cup Z$ and construct a constrained database $\calD$ with three relations: (i) a $4$-ary relation $F(t_0,u_0,t_1,t_\exists)$, (ii) a binary relation $S(u_1,u_\exists)$, and another (iii) binary-relation $C(t_2,u_2)$.
	Moreover, we consider a collection $B$ of FDs and IDs over the schema of $\calT$ using a set $V \dfn \{t_j,u_j \mid 0\leq j\leq 2 \}\cup \{t_\exists,u_\exists\}$ of  {attributes}.
	$B$ contains one FD $\depas{t_0}{u_0}$ similar to the proof of Theorem~\ref{thm:cred-both-fixed}. 
	Due to the change in database relations, we adapt the collection of IDs and consider the set $\{S[u_1]\subseteq C[u_2], C[u_2]\subseteq C[t_2],  C[t_2]\subseteq F[t_1]\}$ to encode whether each clause $c_i\in \varphi$ is satisfied.
	Moreover, we require an additional ID $F[t_{\exists}]\subseteq S[u_{\exists}]$ to encode the existentially quantified variables $Z$.
	As a result, we obtain $B= \{\depas{t_0}{u_0}\}\cup \{ S[u_1]\subseteq C[u_2], C[u_2]\subseteq C[t_2], C[t_2]\subseteq F[t_1], F[t_{\exists}]\subseteq S[u_{\exists}]\}$ as our fixed collection of ICs.
	
	For the database relations, we let $\dom(F)\subseteq\{c_i \mid i\leq m\}\cup X\cup\{0,1,\text{exists}, d\}$,  $\dom(C)=\{c_i \mid i\leq m\}\cup\{d\} $ and $\dom(S)=\{c_1, \text{exists}, d\}$.
	Then our encoding works mostly similar to that in the proof of Theorem~\ref{thm:cred-both-fixed}, except for the minor differences outlined below.
	\begin{enumerate}
		\item We include an $S$-fact to trigger a non-empty set of $C$-facts in any repair. We obtain this by adding the fact $S(c_1, \text{exists})$ together with 
		the fact $C(c_{m}, c_1)$.
		Here, we do not need the fact $C(\text{sat},c_1)$ as in the proof of Theorem~\ref{thm:cred-both-fixed}. 
		The remaining $C$-facts $\{C(c_i, c_{i+1})\mid i< m\}$ are added as before.
		\item We use $s^d\dfn S(d,d), y^d\dfn F(d,d,d,d)$ and $c^d\dfn C(d,d)$ as dummy facts to simulate the effect that any assignment over $Y$ yields a collection of facts in $\calT$ that \emph{trivially} satisfies all the ICs. This is required to simulate the $\forall$ quantification over $Y$.
		\item 
		If a literal $\ell \in \literal(Y)$ does not appear in even a single clause, we still add its corresponding $F$-fact following the same intuition 
		as before, i.e., for such a literal $\ell$: we add 
		$F(y, 1, d, d)$
		if $\ell =y$ and $F(y, 0, d, d)$ if $\ell = \neg y$.
		Moreover, the dummy value $d$ in this case encodes that $\ell$ does not satisfy any clause.
		As a result of this, we can connect repairs to all assignments over universally quantified variables in $Y$.
		For the ease of notation, we write $\ell \in c_0$ for a literal over $Y$ that does not appear in any clause.
		\item The attributes $\{t_{\exists}, u_{\exists}\}$ encode that the $S$-fact $s_{\text{sat}}\dfn S(c_1,\text{exists})$ \emph{supports} facts involving propositions $z\in Z$ via the inclusion dependency $F[t_{\exists}] \subseteq S[u_{\exists}]$.
		This is achieved by letting facts
		$F(z,*,*,\text{exists})$
		for each $z \in Z$ together with the fact $s_{\text{sat}}$.
		However, all facts corresponding to variables  $y\in Y$ take the form $F(y,*,*,d)$ for the dummy value $d$.
		
	\end{enumerate}
	
	Observe that the construction in (3) and (4) requires us to distinguish facts based on variables in $Y$ and $Z$.
	Formally, our database consists of the following facts:
	\begin{align*}
		\calT  = \quad & \{F(y,1,c_i,d) \mid y \in Y \cap c_i, c_i\in\varphi\}\cup \{F(y,0,c_i,d) \mid y\in Y, \neg y \in c_i, c_i\in\varphi\} \\
		& \cup \{F(y,1,d,d) \mid \text{ no clause in } \varphi \text{ contains the literal } y\in Y\} \\
		& \cup \{F(y,0,d,d) \mid \text{ no clause in } \varphi \text{ contains the literal } \neg y \text{ for } y\in Y\} \\
		&\cup ´ \{F(z,1,c_i,\text{exists}) \mid z \in Z\cap c_i, c_i\in\varphi\}\cup \{F(z,0,c_i,\text{exists}) \mid z\in Z, \neg z \in c_i, c_i\in\varphi\} \\
		& \cup  \{F(d,d,d,d)\} 
		\cup \{S(c_1,\text{exists}), S(d,d))\} \\
		& \cup \{C(c_i, c_{i+1}) \mid 1\leq i< m\} \cup 
		\{C(c_m, c_1), C(d,d)\}.
	\end{align*}
	
	For correctness, notice that every assignment $I_Y$ over $Y$ (seen as a subset of $Y$) corresponds to a set $P_Y =\{ F(y,1,c_i, d) \mid y\in I_Y \text{ and } y\in c_i \text{ for } 0\leq i\leq m \} \cup \{F(y,0,c_i, d) \mid y \not \in I_Y \text{ and } \neg y\in c_i \text{ for } 0\leq i\leq m\}$ of $F$-facts in $\calT$.
	Moreover, $P_Y\models \depas{t_0}{u_0}$ since $I_Y$ is an assignment over $Y$ and thus for each $i\neq j\leq m$, $P_Y$ includes facts of the form either $F(y,1,c_i, d)$ or $F(y,0,c_j, d)$ but not both, for each variable $y\in Y$ and clause $c_i, c_j$.
	Furthermore, it is easy to observe that taking the dummy facts, i.e., $P_D = \{y^d, s^d, c^d\}$ the collection $P_Y\cup P_D$ satisfies the IDs $S[u_1]\subseteq C[u_2]$ and  $C[u_2]\subseteq C[t_2]$.
	Now, we let $\calP_Y = P_Y\cup P_D$.
	We have that $s_{\text{sat}} \not \in \calP_Y$.
	In order to extend $\calP_Y$ by adding $F(z,1,c_i, \text{exists})$ or $F(z,0,c_i, \text{exists})$ for any $z\in Z$ and $i\leq m$, $s_{\text{sat}}$ must be added as well due to the ID $F[t_{\exists}]\subseteq S[{u_{\exists}}]$.
	However, in order to include $s_{\text{sat}}$ in any repair, we have to find $I_Y$ and $I_Z$ that together satisfy $\varphi$ due to the remaining IDs, in particular due to $C[t_2]\subseteq F[t_1]$. 
	As a result, for any interpretation $I_Y$ over $Y$, there is an interpretation $I_Z$ over $Z$ such that: $I_Y \cup I_Z \models \varphi$ if and only if $\calP_Y$ is not a repair for $\calD$ (since it can be extended by adding facts $P_Z =\{ F(z,1,c_i, d) \mid z\in I_Z \text{ and } z\in c_i \text{ for } 0\leq i\leq m \} \cup \{F(z,0,c_i, d) \mid z \not \in I_Z \text{ and } \neg z\in c_i \text{ for } 0\leq i\leq m\}$ for such an interpretation $I_Z$ in this case). 
	Equivalently, there is a repair for $\calD$ not containing $s_{\text{sat}}$ if and only if the formula $\Phi$ is false.
	We conclude by observing that $\Phi$ is true if and only if every repair for $\calD$ contains $s_{\text{sat}}$ if and only if $\allrepair$ is true for $s_{\text{sat}}$.
\end{proof}

We provide an example for better understanding of the reduction from the proof of Theorem~\ref{thm:skep-both}.

\begin{example}\label{ex2:red-both}
	Let $\Phi = \forall x\forall y \exists z\exists w((x \lor y \lor z) \land (y \lor \neg z \lor \neg w) \land (y \lor z \lor w))$ be a $\tqbf$.
	Then, our reduction yields a database $\calT$ with three database tables as depicted in Table~\ref{tab2:red-both} and the (fixed) collection 
	$B\dfn \{\depas{t_0}{u_0}\}\cup \{ S[u_1]\subseteq C[u_2], C[u_2]\subseteq C[t_2], C[t_2]\subseteq F[t_1], F[t_{\exists}]\subseteq S[u_{\exists}]\}$ of dependencies.
	The reader can verify that the formula $\Phi$ is true and that for $s_\text{sat}$, $\allrepair$ is true as well.
	For instance, the assignment $\{x\mapsto0, y\mapsto 0\}$ results in $I_Y=\emptyset$ and thus $P_Y= \{\bar x_0, \bar y_0\}$.
	Then, although $P_Y\cup\{y^d, s^d, c^d\}$ satisfies each dependency in $B$, it is not a repair since it is not subset-maximal.
	A repair is obtained by adding further the facts $\{z_1,z_3, \bar w_2\}\cup \{S,s_1,s_2,s_3\}$.
	\begin{table}[t]
		\centering
		\begin{tabular}{l@{\hskip5pt} |@{\hskip5pt} c @{\hskip5pt}| @{\hskip5pt}c}
			$S$& $u_1$  & $u_\exists$  \\\hline
			$s_{\text{sat}}$ & $c_1$  & exists \\
			$s^d$ & $d$ & $d$ \\
		\end{tabular}
		\hspace{1cm}
		\begin{tabular}{l@{\hskip5pt}|@{\hskip5pt} cc@{\hskip5pt}|@{\hskip3pt}c @{\hskip5pt}|@{\hskip3pt} c}
			$F$ & $t_0$ & $u_0$  & $t_1$  & $t_\exists$ \\ \hline 
			$y^d$
			& $d$ & $d$  & $d$ & $d$  \\ \hline 
			$x_{1}$ & $x$ & 1  & $c_1$  & $d$\\ 
			$\bar x_{0}$ & $x$ & 0  & $d$  & $d$\\ \hline 
			$y_{1}$ & $y$ & 1  & $c_1$ & $d$ \\
			$y_{2}$ & $y$ & 1  & $c_2$ & $d$ \\ 
			$y_{3}$ & $y$ & 1  & $c_3$ & $d$ \\ 
			$\bar y_{0}$ & $y$ & $0$  & $d$ & $d$\\ \hline 
			$z_{1}$ & $z$ & 1  & $c_1$ & exists \\
			$\bar z_{2}$ & $z$ & $0$ & $c_2$ & exists \\
			$z_{3}$ & $z$ & 1  & $c_3$ & exists \\ \hline
			$\bar w_{2}$ & $w$ & $0$ & $c_2$ & exists\\ 
			$w_{3}$ & $w$ & 1  & $c_3$ & exists \\
		\end{tabular}
		\hspace{1cm}
		\begin{tabular}{l@{\hskip5pt} |@{\hskip5pt} c @{\hskip5pt}| @{\hskip5pt}c}
			$C$& $t_2$  & $u_2$  \\\hline
			$s_1$ & $c_1$  & $c_2$ \\
			$s_2$ & $c_2$  & $c_3$ \\ 
			$s_3$ & $c_3$  & $c_1$ \\
			$c^d$ & $d$ & $d$ \\
		\end{tabular}
		\caption{The database corresponding to the $\tqbf$ instance $\Phi$ from Example~\ref{ex2:red-both}.}
		\label{tab2:red-both}
	\end{table}
\end{example}

\section{Concluding Remarks}
\paragraph{Overview.}
We 
simulated the problem of finding repairs of an inconsistent database under various families of ICs by Dung's (set-based) argumentation frameworks.
Our main results (see Table~\ref{table:cont}) indicate that subset-maximal repairs correspond to naive extensions when only one type of dependencies are allowed, whereas only preferred extensions yield all the repairs when both types are allowed.
Note that the case of LTGDs (or IDs) requires a polynomial time pre-processing to eliminate all the non-accepted arguments from the resulting (SET)AFs. 
The result of which yields a unique preferred extension which is also stable and naive.

Moreover, the following interesting facts can be derived: ``FDs can be modeled via AFs'' in a similar way in which ``DCs can be modeled via SETAFs''.
Analogous results hold when we replace FDs and DCs by (1) IDs  and LTGDs, respectively, or (2) FDs + IDs and DCs + LTGDs, respectively. 
Furthermore, for the problem to determine whether a tuple is in some (resp., every) repair, we establish the same complexity bounds as the complexity of credulous (skeptical) reasoning for preferred semantics in AFs.
Interestingly, the combined complexity with FDs+IDs remains the same as the data complexity with DCs+LTGDs for considered problems.

\subsection{The Role of Self-attacking Arguments}\label{sec:self-attacks}

Observe that our construction of (SET)AFs for LTGDs and IDs relies on self-attacking auxiliary arguments.
Our motivation lies in avoiding those arguments in any extension of the resulting (SET)AF and connecting repairs to extensions directly.
In the absence of self-attacking arguments, extensions may contain auxiliary arguments and therefore a precise connection between repairs and extensions cannot be established (see Example~\ref{without-selfloops}).

We find it worth highlighting that a considerable thought has been given on how LTGDs/IDs can be modeled in an argumentation framework. 
Intuitively, the semantics of inclusion dependencies closely resemble the notion of ``necessary support'' in the argumentation literature~\cite{AmgoudCLL08,NouiouaR11,DBLP:conf/aaai/0002HN25}. 
If a DB instance contains only one inclusion dependency ($i$) per relation, we can simply model the set of supporting facts in $\support{i}{s}$ for a fact $s$ via the \emph{support} relation between arguments.
However, we observed that this does not apply to the case when an instance $\calD$ contains multiple IDs with the same relation as their source. 
Modeling this requires as many (distinct) support relations as the number of IDs.
This holds since for each fact $s$, one would have to distinguish the satisfaction of individual IDs (i.e., the sets $\support{i}{s}$ and $\support{j}{s}$ for IDs $i\neq j$) which a single support relation cannot offer.
Our translation captures the same essence of supporting arguments via auxiliary attackers, with the benefit that one can distinguish supporters and attackers corresponding to each inclusion dependency. 
As a future work, we aim to think further in this direction and explore ideas that might lead to improved translations for IDs and LTGDs.

\subsection{Discussion and Future Work}
We would like to point out that, although a correspondence between subset-maximal repairs in the presence of functional dependencies (resp., denial constraints) and extensions for (SET)AFs is known~\cite{BienvenuB20}, the main contributions of our work establish the correspondence when inclusion dependencies or LTGDs are also allowed.
This novel contribution opens up several directions for future work.
First and foremost, the authors believe that the connection between repairs in the setting of inconsistent databases and extensions in AFs is stronger than what is established here.
Intuitively, one can model the attack relationship via functional dependencies and defense/support via inclusion dependencies.
A transformation based on this intuition was recently presented by \cite{DBLP:conf/aaai/0002HN25}, showing that AFs can be simulated via inconsistent databases considering FDs and IDs.

Further future work may consider whether the connection between inconsistent databases and argumentation can be generalized to other well-known types of expressive ICs.
This is particularly interesting keeping in mind that the current paper extends an earlier work for FDs and IDs~\cite{mahmood2024computing} to certain expressive ICs towards the \emph{both side} of Figure~\ref{fig:ICs}.
Currently, the authors believe that to model arbitrary (or full) tuple-generating dependencies, one needs the so-called hyper argumentation frameworks~\cite{dimopoulos2023sets}, allowing attacks between two sets of arguments. 
Furthermore, one can also target the richer setting of universal constraints and the symmetric difference repairs~\cite{BienvenuB23}.
%
%
Another interesting question is to explore whether the lower bounds for complexity also apply to the sub-classes of FDs and IDs, namely \emph{keys} and \emph{foreign keys} as well as the particular case of acyclic dependencies.
Moreover, we would like to explore whether consistent query answering (CQA) under inconsistency-tolerant semantics can also be tackled via the argumentation approach.
Finally, one can consider incorporating information about priorities among tuples into the resulting AFs, that is, extending the translations presented in this work to the setting of prioritized repairing and consistent query answering~\cite{FaginKK15,kimelfeld-17,KimelfeldLP20}. 
Here, \cite{BienvenuB20} has already considered FDs and DCs, thus the question remains open only for the case of IDs and LTGDs.

Another promising direction to consider next is the exploration of an explainability dimension, similar to that in the setting of ontological KBs~\cite{AriouaC16,AriouaTC15,bienvenu2019computing}.
Given an instance $\calD$ including a database $\calT$ and a collection $B$ of dependencies, then the proposed AF $\AF{F}{D}$ lets one determine the \emph{causes} why some tuples are not in some repair (or all repairs).
We note that, for IDs, the auxiliary arguments modelling each dependency in $B$ can serve this purpose.
For FDs, we believe that annotating arguments (or the attack relation between a pair of arguments) by the FDs involved in the conflict can achieve the goal.
As a result, one can look at the AF $\AF{F}{D}$ and read from it the FDs or IDs which a tuple $s$ failing $\somerepair$ or $\allrepair$ participates in.
Then, subsets of the atoms and/or possibly tuples in a database can be considered as \emph{explanations}.
Such explanations seem interesting in modeling scenarios where the data (database) has higher confidence than the dependencies; for example, if dependencies are mined over some part of the existing data.
An explanation then informs that the data (and hence tuples therein) should be kept, whereas dependencies need to be screened and further analyzed. 


%
\section*{Acknowledgments}
This is a pre-print to the paper accepted at the Knowledge Engineering Review journal.
Research was funded by the  German Research Foundation (DFG), grant TRR 318/3 2026 – 438445824 and VI 1045-1/1 - 432788559, 
the Ministry of Culture and Science of North Rhine-Westphalia (MKW NRW) within projects WHALE (LFN 1-04) funded under the Lamarr Fellow Network programme and
project SAIL, grant NW21-059D,
and by the German Federal Ministry of Research, Technology and Space (BMFTR) within the project KI-Akademie OWL under the grant no 16IS24057B.

\bibliographystyle{abbrv}
\bibliography{main}

\end{document}